\documentclass[letterpaper,twocolumn,10pt]{article}
\usepackage{usenix19_v3}

\newif\ifdraft
\newif\ifcomments
\newif\iffull
\newif\ifpadding
\newif\ifextended

\draftfalse
\commentsfalse
\fullfalse
\paddingtrue
\extendedtrue

\usepackage{amsmath}
\usepackage{amsthm}
\usepackage{varwidth}
\usepackage[shortlabels,inline]{enumitem}
\usepackage{subfigure}
\usepackage{xspace}
\usepackage{amssymb}
\usepackage{hhline}
\usepackage{tabularx, booktabs}
\usepackage{algorithmicx}
\usepackage{algpseudocode}
\usepackage{algorithm}
\usepackage{xcolor}
\usepackage{microtype}
\usepackage{multirow}
\usepackage[binary-units]{siunitx}
    \DeclareSIUnit\bps{b/s}
\usepackage[colorlinks=true,
             linkcolor=black,
             urlcolor=black,
             citecolor=black,
             breaklinks, 
             pdfborder={0 0 0}]{hyperref}
\usepackage[aboveskip=5pt,small,textfont=bf,labelfont=bf]{caption}
\usepackage[normalem]{ulem} %
\usepackage{inconsolata}
\usepackage{comment}
\usepackage[noabbrev,capitalize,nameinlink]{cleveref}

\ifcomments
\usepackage[textsize=small,color=lightgray,colorinlistoftodos]{todonotes}
\else
\usepackage[disable]{todonotes}
\fi

\ifdraft
    \usepackage{fancyhdr, datetime}
\fi
\usepackage[small,compact]{titlesec}
\titleformat{\subsubsection}{\normalfont\em\bf}{\thesubsubsection}{1em}{\em}
\titlespacing*{\section}{0pt}{7pt}{3pt}
\titlespacing*{\subsection}{0pt}{4pt}{2pt}
\titlespacing*{\subsubsection}{0pt}{2pt}{0pt}

\DeclareMathAlphabet{\mathcal}{OMS}{cmsy}{m}{n}

\renewcommand{\paragraph}[1]{\smallskip\noindent\textbf{#1.}\enskip}

\newcommand{\code}[1]{\texttt{#1}\xspace}

\definecolor{FGreen}{cmyk}{0.9,0.2,0.5,0.3}

\makeatletter
\newcommand*{\ie}{%
    \@ifnextchar{,}%
        {i.e.}%
        {i.e.,\@\xspace}%
}
\makeatletter
\newcommand*{\eg}{%
    \@ifnextchar{,}%
        {e.g.}%
        {e.g.,\@\xspace}%
}
\makeatletter
\newcommand*{\etc}{%
    \@ifnextchar{.}%
        {etc}%
        {etc.\@\xspace}%
}
\makeatother
\makeatletter
\newcommand*{\etal}{%
    \@ifnextchar{.}%
        {\textit{et al}}%
        {\textit{et al.}\@\xspace}%
}
\makeatother

\crefformat{section}{\S#2#1#3}

\def\approx{$\mathtt{\sim}$} %
\newcommand*\bubble[1]{\tikz[baseline=(char.base)]{
            \node[shape=circle,scale=0.8,draw,inner sep=2pt,fill=black,
            text=white] (char)
            {#1};}}

\makeatletter
\def\thickhline{%
  \noalign{\ifnum0=`}\fi\hrule \@height \thickarrayrulewidth \futurelet
   \reserved@a\@xthickhline}
\def\@xthickhline{\ifx\reserved@a\thickhline
               \vskip\doublerulesep
               \vskip-\thickarrayrulewidth
             \fi
      \ifnum0=`{\fi}}
\makeatother
\newlength{\thickarrayrulewidth}
\newcommand\T{\rule{0pt}{2.6ex}}       %
\newcommand\B{\rule[-1.2ex]{0pt}{0pt}} %

\def\oselect{\code{oassign}}
\def\osort{\code{osort}}
\def\oarr{\code{oaccess}}
\def\dummy{\code{isDummy}}
\def\mynull{\code{null}}

\newenvironment{compactenumerate}{ 
    \begin{enumerate}[nolistsep,leftmargin=*,label=\arabic*)]
        \setlength{\itemsep}{0.1pt}
        \setlength{\parskip}{0.1pt}
        \setlength{\parsep}{0.1pt}     
    }{\end{enumerate}}

\newenvironment{compactenumeratealph}{ 
	\begin{enumerate}[a),nolistsep,leftmargin=*]
		\setlength{\itemsep}{0.1pt}
		\setlength{\parskip}{0.1pt}
		\setlength{\parsep}{0.1pt}     
	}{\end{enumerate}}

\newtheorem{theorem}{Theorem}
\newtheorem{definition}{Definition}

\newcommand{\eat}[1]{\ignorespaces}
\newcommand{\drop}[1]{\ignorespaces}

\usepackage{letltxmacro}
\LetLtxMacro{\todonote}{\todo}
\usepackage{setspace}

\renewcommand{\todo}[2][]
{\todonote[size=\small,caption={#2}, #1]
{\begin{spacing}{0.5}#2\end{spacing}}}

\ifcomments
\paperwidth=\dimexpr \paperwidth + 10cm\relax
\oddsidemargin=\dimexpr\oddsidemargin + 5cm\relax
\evensidemargin=\dimexpr\evensidemargin + 5cm\relax
\marginparwidth=\dimexpr \marginparwidth + 5cm\relax
\fi

\newcommand{\sys}{Visor\xspace}

\begin{document}

\date{}

 \twocolumn[
 \bigskip
\centerline{\Large \bf \sys: Privacy-Preserving Video Analytics as a Cloud Service}
 \bigskip
 \centerline{Rishabh Poddar$^{1, 2}$, Ganesh Ananthanarayanan$^2$, Srinath Setty$^2$, Stavros Volos$^2$, Raluca Ada Popa$^1$}
 \smallskip
 \centerline{$^1$UC Berkeley \quad $^2$Microsoft Research}
 \smallskip
 \centerline{\small \code{<rishabhp,raluca>@eecs.berkeley.edu} \quad \code{<ga,srinath,svolos>@microsoft.com}}
 \bigskip %
 ]

\newcommand{\finaloverhead}{{$2\times$--$6\times$}\xspace}
\newcommand{\finaloverheadnopadding}{{$1.6\times$--$2.9\times$}\xspace}

\begin{abstract}
   Video-analytics-as-a-service is becoming an important offering for cloud providers. A key concern in such services is privacy of the videos being analyzed. While trusted execution environments (TEEs) %
   are promising options for preventing the direct leakage of private video content, they remain vulnerable to {\em side-channel attacks}. %

   We present \sys, a system that provides confidentiality for the user's video stream as well as the ML models %
   in the presence of a compromised cloud platform and untrusted co-tenants. 
   \sys executes video pipelines in a \emph{hybrid} TEE that spans both the CPU and GPU. It protects the pipeline against side-channel attacks induced by data-dependent access patterns of video modules, and also addresses leakage in the CPU-GPU communication channel.  
    \sys is up to $1000\times$ faster than na\"ive oblivious solutions, and its overheads relative to a {\em non-oblivious baseline} are limited to \finaloverhead.
\end{abstract}

\pagestyle{plain}
\pagenumbering{gobble}

\section{Introduction}

Cameras are being deployed pervasively %
for 
the many applications they enable, such as traffic planning, retail experience, and enterprise security \cite{VisionZero,camera:retail,Verkada}. Videos from the cameras are streamed to the cloud, where they are processed using video analytics pipelines \cite{VideoStorm, Chameleon:Sigcomm, NoScope} composed of computer vision techniques (e.g., OpenCV \cite{opencv}) and convolutional neural networks (e.g., object detector CNNs \cite{Redmon:YOLO}); as illustrated in \cref{fig:pipeline}. Indeed, ``video-analytics-as-a-service'' is becoming an important offering for cloud providers \cite{Azure:Video:analytics, Amazon:Video:analytics}.

Privacy of the video contents is of paramount concern in the ``video analytics-as-a-service'' offerings. Videos often contain sensitive information, such as users' home interiors, people in workspaces, or license plates of cars. %
For example, the Kuna home monitoring service~\cite{Kuna} transmits videos from users' homes to the cloud, analyzes the videos, and notifies users when it detects movement in areas of interest. 
For user privacy, video streams must remain {\em confidential} and not be revealed to the cloud provider %
or other co-tenants in the cloud.

Trusted execution environments (TEEs)~\cite{SGX:HASP13,Graviton:OSDI18} are a natural fit for privacy-preserving video analytics in the cloud. In contrast to cryptographic approaches, such as homomorphic encryption, TEEs rely on the assumption that cloud tenants also trust the hardware. %
The hardware provides %
the ability to create secure ``enclaves'' that are protected against privileged attackers. TEEs are more compelling than cryptographic techniques since they are orders of magnitude faster. 
In fact, CPU TEEs (e.g., Intel SGX \cite{SGX:HASP13}) lie at the heart of confidential cloud computing \cite{IBM:SGX, Azure:SGX}. Meanwhile, recent advancements in GPU TEEs \cite{Graviton:OSDI18,HIX:Jang:2019} enable the execution of ML models (e.g., neural networks) with strong privacy guarantees as well. CPU and GPU TEEs, thus, present an opportunity for building privacy-preserving video analytics systems.

Unfortunately, TEEs (e.g., Intel SGX) are vulnerable to a host of side-channel attacks %
(e.g., \cite{wang-sgx-leaky, sgxcache-brasser, sgxattacks-foreshadow, sgxattacks-xu:pagefaults}). %
For instance, in \cref{s:background:attacks} we show that by observing just the memory access patterns of a widely used bounding box detection OpenCV module, an attacker can infer the {\em exact shapes and positions of all moving objects} in the video.
In general, an attacker can infer crucial information about the video being processed, such as the times when there is activity, objects that appear in the video frame, all of which when combined with knowledge about the physical space being covered by the camera, can lead to serious violations of confidentiality.

We present \sys, a system for privacy-preserving video analytics services.
\sys protects the confidentiality of the videos being analyzed from the service provider and other co-tenants.
When tenants host their own CNN models in the cloud, it also protects the model parameters and weights.  %
\sys protects against a powerful enclave attacker who can compromise the software stack outside the enclave, as well as observe any {\em data-dependent accesses} to network, disk, or memory via side-channels (similar to prior work \cite{Ohrimenko:ObliviousML, Raccoon}). %

\sys makes two primary contributions,
combining insights from ML systems, security, computer vision, and algorithm design.
First, we present a privacy-preserving framework for machine-learning-as-a-service (MLaaS), which supports CNN-based ML applications spanning both CPU and GPU resources.
Our framework can potentially power applications beyond video analytics, such as medical imaging, recommendation systems, and financial forecasting.
Second, we develop novel \emph{data-oblivious} algorithms with provable privacy guarantees within our MLaaS framework, for commonly used vision modules.
The modules are efficient and can be composed to construct many different video analytics pipelines. 
In designing our algorithms, we formulate a set of design principles that can be broadly applied to other vision modules as well.

\paragraph{1) Privacy-Preserving MLaaS Framework}
\sys leverages a \emph{hybrid} TEE that spans CPU and GPU resources available in the cloud. Recent work has shown that scaling video analytics pipelines requires judicious use of both CPUs and GPUs~\cite{Scanner:Poms:2018,Focus:OSDI18}. 
Some pipeline modules can run on CPUs at the required frame rates (\eg video decoding or vision algorithms) while others (\eg CNNs) require GPUs, as shown in Figure \ref{fig:pipeline}. 
Thus, our solution spans both CPU and GPU TEEs, and combines them into a unified trust domain.

\sys systematically addresses access-pattern-based leakage across the components of the hybrid TEE, from video ingestion to CPU-GPU communication to CNN processing. In particular, we take the following steps:
\begin{compactenumeratealph}
\item \sys leverages a suite of data-oblivious primitives to remove access pattern leakage from the CPU TEE. The primitives enable the development of oblivious modules with provable privacy guarantees, the access patterns of which are always independent of private data. 

\item \sys relies on a novel oblivious communication protocol to remove leakage from the CPU-GPU channel. As the CPU modules serve as filters, the data flow in the CPU-GPU channel (on which objects of each frame are passed to the GPU) leaks information about the contents of each frame, enabling attackers to infer the number of moving objects in a frame. At a high level, \sys pads the channel with dummy objects, leveraging the observation that our application is not constrained by the CPU-GPU bandwidth. To reduce GPU wastage, \sys intelligently minimizes running the CNN on the dummy objects.

\item \sys makes CNNs running in a GPU TEE oblivious by leveraging \emph{branchless} CUDA instructions to implement conditional operations (e.g., ReLU and max pooling) in a data-oblivious way.

\end{compactenumeratealph}

\paragraph{2) Efficient Oblivious Vision Pipelines}
Next, we design novel data-oblivious algorithms for vision modules that are foundational for video analytics, and implement them using the oblivious primitives provided by the framework described above.
Vision algorithms are used in video analytics pipelines to extract the moving foreground objects. These algorithms (\eg background subtraction, bounding box detection, object cropping, and tracking) run on CPUs and serve as cheap {\em filters} to discard frames instead of invoking expensive CNNs on the GPU for each frame's objects (more in \cref{s:model}). %
The modules can be composed to construct various vision pipelines, such as medical imaging and motion tracking.

As we demonstrate in \cref{s:evaluation}, na\"ive approaches for making these algorithms data-oblivious, such that their operations are independent of each pixel's value, can slow down video pipelines by several orders of magnitude. 
Instead, we carefully craft oblivious vision algorithms for each module in the video analytics pipeline, including the popular VP8 video decoder \cite{VP8Overview}.
Our overarching goal is to transform each algorithm into a pattern that processes each pixel identically.
To apply this design pattern efficiently, we devise a set of algorithmic and systemic optimization strategies based on the properties of vision modules, as follows.
First, we employ a divide-and conquer approach---\ie we break down each algorithm into independent subroutines based on their functionality, and tailor each subroutine individually.
Second, we cast sequential algorithms into a form that \emph{scans} input images while performing identical operations on each pixel. %
Third, identical pixel operations allow us to systemically amortize the processing cost across groups of pixels in each algorithm.
For each vision module, we derive the operations applied per pixel in conjunction with these design strategies.
Collectively, these strategies improve performance by up to $1000\times$ over na\"ive oblivious solutions.
We discuss our approach in more detail in \cref{s:obl_overview}; %
nevertheless, we note that it can potentially help inform the design of other oblivious vision modules as well, beyond the ones we consider in \sys.

In addition, as shown by prior work, bitrate variations in encrypted network traffic can also leak information about the underlying video streams~\cite{Schuster:Video:Attack}, beyond access pattern leakage at the cloud.
To prevent this leakage, we modify the video encoder to carefully pad video streams \emph{at the source} in a way that optimizes the video decoder's latency.
\sys thus provides an end-to-end solution for private video analytics.

\paragraph{Evaluation Highlights}
We have implemented \sys on Intel SGX CPU enclaves \cite{SGX:HASP13} and Graviton GPU enclaves \cite{Graviton:OSDI18}.
We evaluate \sys on commercial video streams of cities and datacenter premises containing sensitive data. 
Our evaluation shows that \sys's vision components perform up to $1000\times$ better than na\"ive oblivious solutions, and over $6$ to $7$ orders of magnitude better than a state-of-the-art general-purpose system for oblivious program execution.
Against a {\em non-oblivious baseline}, \sys's overheads are limited to \finaloverhead which still enables us to analyze multiple streams simultaneously in real-time on our testbed.
\sys is versatile and can accommodate different combinations of vision components used in real-world applications.
Thus, \sys provides an efficient solution for private video analytics.

\begin{figure}[t!]
    \centering
    \subfigure[\bf Pipeline with object classifier (e.g., ResNet).]{
        \includegraphics[height=2cm]{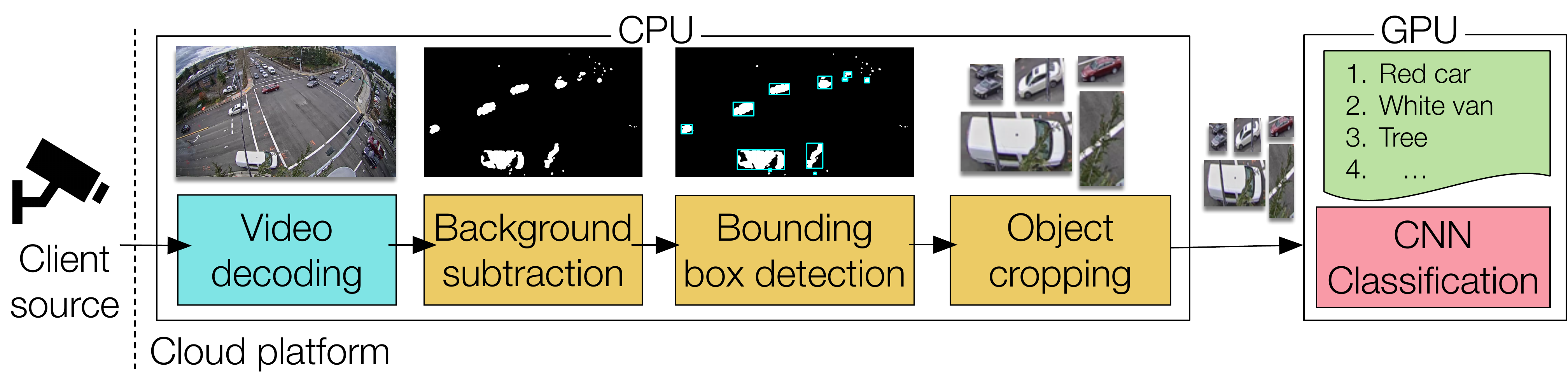}
    \label{fig:pipeline:resnet}}
    \hspace{.25in}
    \subfigure[\bf Pipeline with object detector (e.g., Yolo).]{
        \includegraphics[height=2cm]{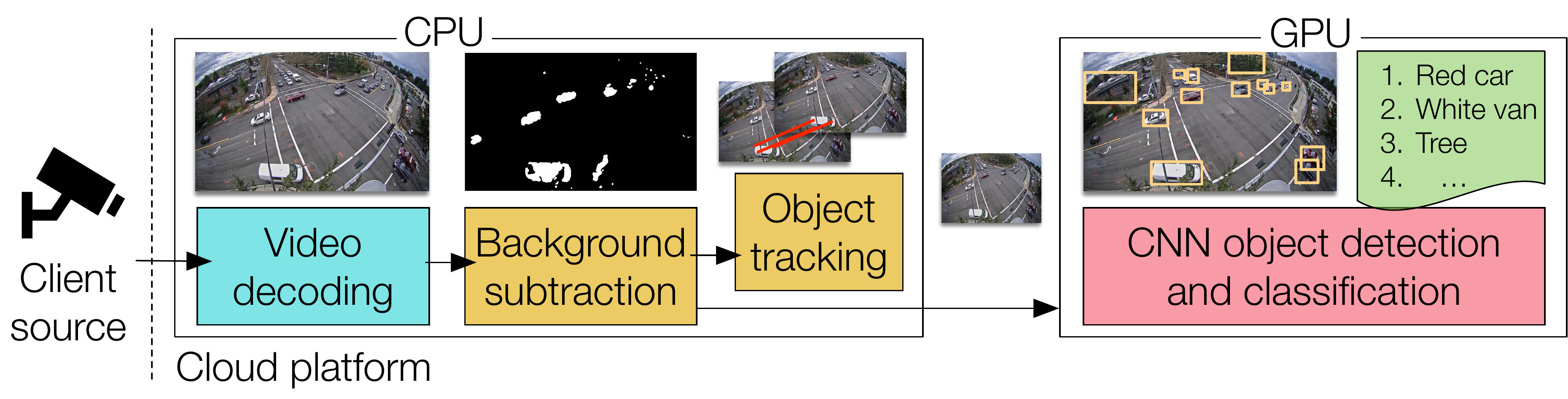}
    \label{fig:pipeline:yolo}}
    \caption{Video analytics pipelines. Pipeline (a) extracts the objects using vision algorithms and classifies the cropped objects using a CNN classifier on the GPU. Pipeline (b) also uses the vision algorithms as a filter, but sends the entire frame to the CNN detector. Both pipelines may optionally use object tracking. %
    }
    \label{fig:pipeline}
\end{figure}

\section{Background and Motivation}\label{s:background}

\subsection{Video Analytics as a Service} \label{s:model}

\cref{fig:pipeline} depicts the canonical pipelines for video analytics~\cite{Focus:OSDI18, NoScope, VideoStorm, AWStream, MS:Rocket:web}.
The client (\eg a source camera) feeds the video stream to the service hosted in the cloud, which (a)~decodes the video into frames, (b)~extracts objects from the frames using vision algorithms, and (c)~classifies the objects using a pre-trained convolutional neural network (CNN). Cameras typically offer the ability to control the resolution and frame rate at which the video streams are encoded. 

Recent work demonstrates that scaling video analytics pipelines requires judicious use of both CPUs and GPUs~\cite{Scanner:Poms:2018,Focus:OSDI18}.
In \sys, we follow the example of Microsoft's Rocket platform for video analytics~\cite{MS:Rocket:web,MS:Rocket:git}---we split the pipelines by running video decoding and vision modules on the CPU, while offloading the CNN to the GPU (as shown in \cref{fig:pipeline}).
The vision modules process each frame to detect the moving ``foreground'' objects in the video using background subtraction~\cite{MOG:Survey:2008}, compute each object's bounding box \cite{Suzuki:1985}, and crop them from the  frame for the CNN classifier. These vision modules can sustain the typical frame rates of videos even on CPUs, thereby serving as vital ``filters'' to reduce the expensive CNN operations on the GPU \cite{Focus:OSDI18, NoScope}, and are thus widely used in practical deployments. For example, CNN classification in \cref{fig:pipeline:resnet} is invoked only if moving objects are detected in a region of interest in the frame. 
Optionally, the moving objects are also tracked to infer directions (say, cars turning left). 
The CNNs can either be object classifiers (e.g., ResNet \cite{He:Resnet:2016}) as in \cref{fig:pipeline:resnet}; or object detectors (\eg Yolo \cite{Redmon:YOLO}) as in \cref{fig:pipeline:yolo}, which take whole frames as input. 
The choice of pipeline modules is application dependent \cite{Chameleon:Sigcomm, Focus:OSDI18} and \sys targets confidentiality for all pipeline modules, their different combinations, and vision CNNs. %

While our description focuses on a multi-tenant cloud service, our ideas equally apply to multi-tenant {\em edge compute} systems, say, at cellular base stations~\cite{MEC:ETSI}. Techniques for lightweight programmability on the cameras to reduce network traffic (\eg using smart encoders \cite{SmartCodec:Vivotek} or dynamically adapting frame rates \cite{Edgecomputing:video}) are orthogonal to {\sys}'s techniques.

\subsection{Trusted Execution Environments}
\label{s:background:enclaves}

Trusted execution environments, or enclaves, protect application's code and data from all other software in a system.
Code and data loaded in an enclave---CPU and GPU TEEs---can be verified by clients using the \emph{remote attestation} feature.

\noindent\textbf{Intel SGX}
\cite{SGX:HASP13} enables TEEs on CPUs and enforces isolation by storing enclave code and data in a protected memory region called the Enclave Page Cache (EPC). %
The hardware ensures that no software outside the enclave can access EPC contents.

\noindent\textbf{Graviton}
\cite{Graviton:OSDI18} enables TEEs on GPUs in tandem with trusted applications hosted in CPU TEEs. 
Graviton prevents an adversary from observing or tampering with traffic (data and commands) transferred to/from the GPU. 
A trusted GPU runtime (\eg CUDA runtime) hosted in a CPU TEE attests that all code/data have been securely loaded onto the GPU.

\subsection{Attacks based on Access Pattern Leakage}
\label{s:background:attacks}

TEEs are vulnerable to leakage from side-channel attacks %
that exploit micro-architectural side-channels 
\cite{sgxattacks-foreshadow, sgxcache-gotzfried, sgxcache-brasser, sgxcache-schwarz, sgxcache-moghimi, sgxattacks-hahnel:cache, sgxattacks-lee:branches, SGX:attack:ZombieLoad, sgxcache-cachequote}, 
software-based channels~\cite{sgxattacks-xu:pagefaults, bulck-sgxattack:pagefaults}, or application-specific leakage, such as network and memory accesses. %

\begin{figure} [t!]
    \centering
    \includegraphics[width=0.7\linewidth]{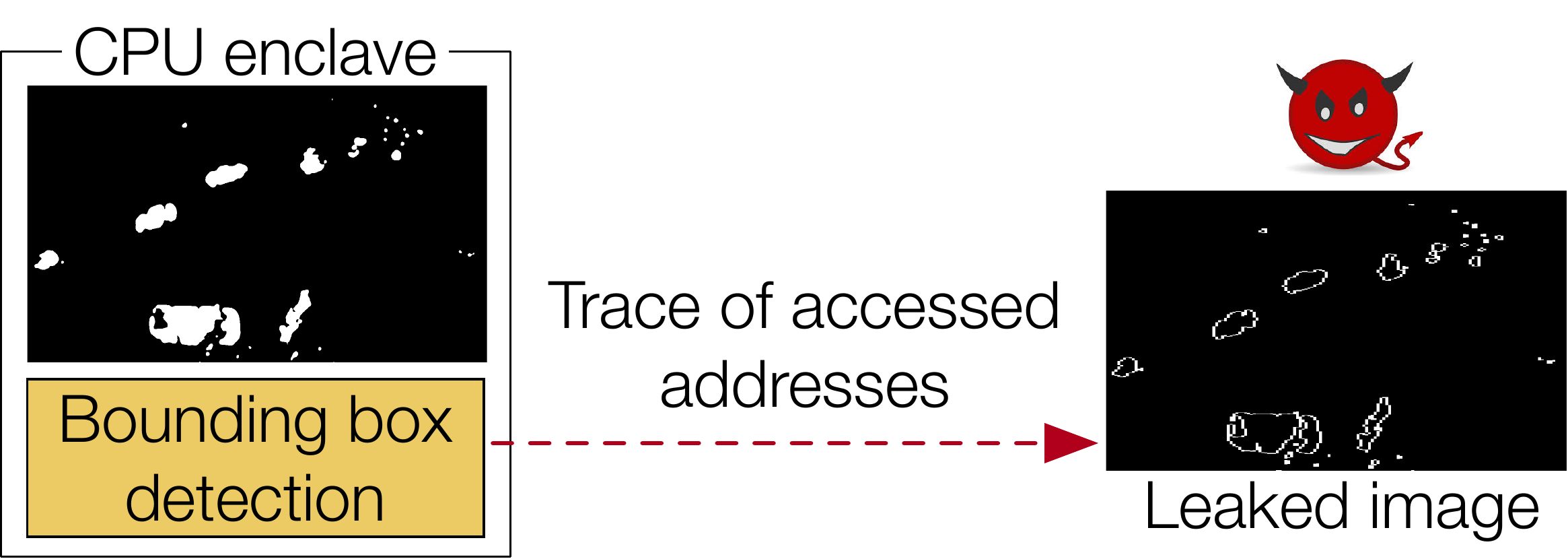}
    \caption{Attacker obtains all the frame's objects (right) using access pattern leakage in the bounding box detection module.}
    \label{fig:leakage}
\end{figure}

A large subset of these attacks exploit {\em data-dependent memory access patterns} %
(\eg %
branch-prediction, cache-timing, or controlled page fault attacks). %
Xu \etal~\cite{sgxattacks-xu:pagefaults} show that by simply observing the page access patterns of image decoders, an attacker can reconstruct entire images.
We ourselves analyzed the impact of access pattern leakage at cache-line granularity~\cite{sgxcache-gotzfried,sgxcache-brasser,sgxcache-schwarz,sgxcache-moghimi} on the bounding box detection algorithm \cite{Suzuki:1985} (see \cref{fig:pipeline:resnet}; \S\ref{s:model}). 
We simulated existing attacks by capturing the memory access trace during an execution of the algorithm, and then examined the trace to reverse-engineer the contents of the input frame.
Since images are laid out predictably in memory, we found that the attacker is able to infer the locations of all the pixels touched during execution, and thus, the {\em shapes and positions of all objects} (as shown in \cref{fig:leakage}).
Shapes and positions of objects are the core content of any video, and allow the attacker to infer sensitive information like times when patients are visiting private medical centers or when residents are inside a house, and even infer if the individuals are babies or on wheelchairs based on their size and shapes. In fact, conversations with customers of one of the largest public cloud providers
indeed confirm that {\em privacy of the videos is among their top-two concerns} in signing up for the video analytics cloud service.%

\section{Threat Model and Security Guarantees}\label{s:threatmodel}

We describe the attacker's capabilities and lay out the attacks that are in scope and out of scope for our work. %

\subsection{Hardware Enclaves and Side-Channels}\label{s:threatmodel:attacks}
Our trusted computing base includes: 
\begin{enumerate*}[($i$)]
	\item the GPU package and its enclave implementation, 
	\item the CPU package and its enclave implementation, and 
	\item the video analytics pipeline implementation and GPU runtime hosted in the CPU enclave.
\end{enumerate*}

The design of \sys is not tied to any specific hardware enclave; instead, \sys builds 
on top of an {\em abstract} model of hardware enclaves where the attacker controls the server’s software stack outside the enclave (including the OS), but cannot perform any attacks to glean information from inside the processor (including processor keys). 
The attacker can additionally observe the contents and access patterns of all (encrypted) pages in memory, for both data and code.
We assume that the attacker can observe the enclave's memory access patterns at cache line granularity~\cite{Ohrimenko:ObliviousML}. 
Note that our attacker model includes the cloud service provider as well as other co-tenants.

We instantiate \sys with the widely-deployed Intel SGX enclave. However, recent attacks show that SGX does not quite satisfy the abstract enclave model that \sys requires. For example, attackers may be able to distinguish \emph{intra} cache line memory accesses~\cite{sgxattacks:CacheBleed,sgxattacks:Memjam}.
In \sys, we mitigate these attacks by disabling hyperthreading in the underlying system, disallowing attackers from observing intra-core side-channels; clients can verify that hyperthreading is disabled during remote attestation~\cite{SGX:IAS}.
One may also employ complementary solutions for closing hyperthreading-based attacks~\cite{SGX:defense:Varys,SGX:defense:Hyperrace}.

Other attacks that violate our abstract enclave model are out of scope:
such as attacks based on timing analysis or %
power consumption~\cite{SGX:attack:Plundervolt,Attack:Clkscrew}, DoS attacks~\cite{SGXattack:rowhammer:SGXbomb:DOS:Jang,SGX:attack:rowhammer}, or rollback attacks~\cite{Memoir:rollback} (which have complementary solutions~\cite{SGX:LCM:defense:rollback:Brandenburger:2018, SGX:ROTE:defense:rollback:Matetic:2017}).
Transient execution attacks (\eg \cite{sgxattacks-foreshadow,SGX:attack:ZombieLoad,sgxattacks-sgxpectre,vanbulck2020lvi,crosstalk,MDS:attack:RIDL,cacheOut}) are also out of scope;
these attacks violate the threat model of SGX and are typically patched promptly by the vendor via microcode updates.
In the future, one could swap out Intel SGX in our implementation for upcoming enclaves such as MI6~\cite{mi6}  and Keystone~\cite{keystone} that address many of the above drawbacks of SGX.

\sys provides protection against {\em any channel of attack that exploits data-dependent access patterns} within our abstract enclave model, which represent a large class of known attacks on enclaves
(\eg cache attacks~\cite{sgxcache-gotzfried, sgxcache-brasser, sgxcache-schwarz, sgxcache-moghimi, sgxattacks-hahnel:cache}, branch prediction~\cite{sgxattacks-lee:branches}, paging-based attacks~\cite{sgxattacks-xu:pagefaults, bulck-sgxattack:pagefaults}, or memory bus snooping~\cite{membuster}).
We note that even if co-tenancy is disabled (which comes at considerable expense), privileged software such as the OS and hypervisor can still infer access patterns (\eg by monitoring page faults), thus still requiring data-oblivious solutions.

Recent work has shown side-channel leakage on GPUs ~\cite{GPU:Covert:Naghibijouybari:2017, GPU:Side:Naghibijouybari:2018, GPU:Timing:Jiang:2017, GPU:Timing:Jiang:2016} including the exploitation of data access patterns out of the GPU. We expect similar attacks to be mounted on GPU {\em enclaves} as video and ML workloads gain in popularity, and our threat model applies to GPU enclaves as well.

\ifpadding
\else
Finally, the TLS traffic from the camera leaks the variation in bitrate of the video to an attacker observing the network. 
Padding the video segments {\em at the camera} addresses this leakage \cite{Schuster:Video:Attack}, and this is complementary to \sys's threat model {\em in the cloud}.
For completeness, we measure the impact of such padding on \sys's performance in \cref{s:padding}.

\fi

\subsection{Video Streams and CNN Model}
Each client owns its video streams, and it expects to protect its video from the cloud and co-tenants of the video analytics service. %
The vision algorithms are assumed to be public.

We assume that the CNN model's architecture is public, but its weights are private and may be proprietary to either the client or the cloud service.
\sys protects the weights in both scenarios within enclaves, in accordance with the threat model and guarantees from \cref{s:threatmodel:attacks};
however, when the weights are proprietary to the cloud service, the client may be able to learn some information about the weights by analyzing the results of the pipeline~\cite{ML:model:predictionAPI:Tramer:2016,ML:inversion:Fredrikson:2015,ML:inversion:pharma:Fredrikson:2014}.
Such attacks are out of scope for Visor.

\drop{
    The CNN model has two deployment scenarios. In both scenarios, we assume that the model's architecture is public. 
    In one scenario, the client owns the model weights and analyzes the videos on the cloud due to the availability of richer compute that is unavailable and expensive to provision on-premise. However, the client wishes to conceal the weights from the cloud provider. 
    In the other scenario, the cloud provider owns the model weights and wishes to conceal anything about the weights \emph{from the clients}, beyond what can be inferred from the model's results. \sys protects both the deployment scenarios for CNNs.
    Protecting against attacks~\mbox{\cite{Abadi:MLdefense:DP, Nasr:MLdefense:Regularization:2018, Iyengar:MLdefense:DP:2019}}  that use a model's results to extract its weights~\mbox{\cite{ML:model:predictionAPI:Tramer:2016, ML:model:reverse:Oh:2018, ML:model:hyperparameters:Wang:2018}} or its training data~\mbox{\cite{ML:data:membership:Shokri:2017, ML:data:membership:Salem:2019, ML:data:remember:Song:2017, ML:inversion:Fredrikson:2015, ML:inversion:pharma:Fredrikson:2014}}, is complementary to \sys. 
}

\ifpadding
    Finally, recent work has shown that the camera's encrypted network traffic leaks the video's bitrate variation to an attacker observing the network~\cite{Schuster:Video:Attack}, which may consequently leak information about the video contents.
    \sys eliminates this leakage by padding the video segments {\em at the camera}, in such a way that optimizes the latency of decoding the padded stream at the cloud~(\S\ref{s:padding}).
\fi

\subsection{Provable Guarantees for Data-Obliviousness} %

\sys provides {\em data-obliviousness}
within our abstract enclave model from \cref{s:threatmodel:attacks}, which guarantees that the memory access patterns of enclave code does not reveal any information about sensitive data.
We rely on the enclaves themselves to provide integrity, along with authenticated encryption.

 We formulate the guarantees of data-obliviousness using the ``simulation paradigm'' \cite{Goldreich2004:Vol1}.  
First, we define a {\em trace of observations} that the attacker sees in our threat model. Then, we define the {\em public information}, \ie information we do not attempt to hide and is known to the attacker. Using these, we argue that there exists a simulator, such that for all videos $V$, when given {\em only} the public information (about $V$ and the video algorithms), the simulator can produce a trace that is indistinguishable from the real trace visible to an attacker who observes the access patterns during \sys's processing of $V$. By ``indistinguishable'', we mean that no polynomial-time attacker can distinguish between the simulated trace and the real trace observed by the attacker. %
The fact that a simulator can produce the same observations as  seen by the  attacker {\em even without knowing the private data in the video stream} implies that the attacker does not learn sensitive data about the video.

In our attacker model, the trace of observations is the sequence of the addresses of memory references to code as well as data, along with the accessed data (which is encrypted).
The public information is all of \sys's algorithms, formatting and sizing information, but not the video data. 
For efficiency, \sys also takes as input some public parameters that represent various upper bounds on the properties of the video streams, \eg, the maximum number of objects per frame, or upper bounds on object dimensions.

\ifextended
In \cref{s:proofs}, we provide a formal definition of data-obliviousness (\cref{def:oblivious}); a summary of public information for each algorithm; and proofs of security along with detailed pseudocode for each algorithm.
Since \sys's data-oblivious algorithms (\cref{s:decoding} and \cref{s:algorithms}) follow an {\em identical sequence of memory accesses} that depend only on public information and are {\em independent} of data content, our proofs are easy to verify. %
\else
We defer a formal treatment of \sys's security guarantees---including the definitions and proofs of security, along with detailed pseudocode for each algorithm---to an extended appendix~\cite{visor:extended}.
In summary, we show that \sys's data-oblivious algorithms (\cref{s:decoding} and \cref{s:algorithms}) follow an \emph{identical sequence of memory accesses} that depend only on public information and are \emph{independent} of data content. %
\fi

\section{A Privacy-Preserving MLaaS Framework}
\label{s:system}

In this section, we present a privacy-preserving framework for machine-learning-as-a-service (MLaaS), that supports CNN-based ML applications spanning both CPU and GPU resources.
Though \sys focuses on protecting video analytics pipelines, our framework can more broadly be used for a range of MLaaS applications such as medical imaging, recommendation systems, and financial forecasting. 

Our framework comprises three key features that collectively enable data-oblivious execution of ML services. First, it protects the computation in ML pipelines using a \emph{hybrid} TEE that spans both the CPU and GPU. 
Second, it provides a secure CPU-GPU communication channel that additionally prevents the leakage of information via traffic patterns in the channel.
Third, it prevents access-pattern-based leakage on the CPU and GPU by facilitating the development of data-oblivious modules using a suite of optimized primitives.

\begin{figure} [t!]
    \centering
    \includegraphics[width=\linewidth]{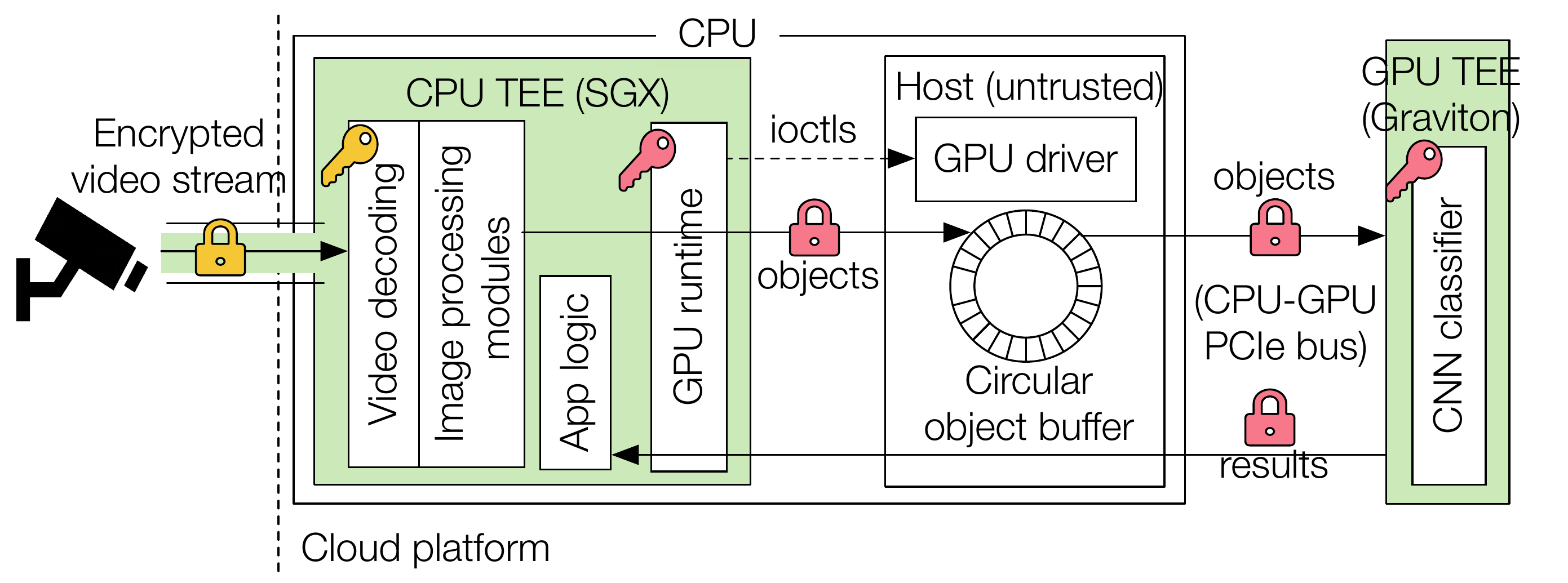}
    \caption{\sys's hybrid TEE architecture. Locks indicate encrypted data channels, and keys indicate decryption points.}
    \label{fig:architecture}
\end{figure}

\subsection{Hybrid TEE Architecture} \label{s:system:architecture}
\cref{fig:architecture} shows \sys's architecture. \sys receives encrypted video streams from the client's camera, which are then fed to the video processing pipeline. We refer to the architecture as a {\em hybrid} TEE as it spans both the CPU and GPU TEEs, with different modules of the video pipeline (\cref{s:model}) being placed across these TEEs.
We follow the example of prior work that has shown that running the non-CNN modules of the pipeline on the CPU, and the CNNs on the GPU \cite{Focus:OSDI18, Scanner:Poms:2018, MS:Rocket:web}, results in efficient use of the expensive GPU resources while still keeping up with the incoming frame rate of videos. %

Regardless of the placement of modules across the CPU and GPU, we note that attacks based on data access patterns can be mounted on {\em both} CPU and GPU TEEs, as explained in \cref{s:threatmodel:attacks}. 
As such, our data-oblivious algorithms and techniques are broadly applicable irrespective of the placement,
though our description is based on non-CNN modules running on the CPU and the CNNs on the GPU.
\drop{
Nonetheless, the placement of modules across the CPU and GPU is orthogonal to \sys. As explained in \cref{s:threatmodel:attacks}, attacks based on data access patterns can be mounted on {\em both} CPU and GPU enclaves, and hence our data-oblivious modules and techniques are applicable regardless of the placement. That said, our description is based on non-CNN modules running on the CPU and the CNNs on the GPU.
}

\paragraph{CPU and GPU TEEs}
We implement the CPU TEE using Intel SGX enclaves, and the GPU TEE using Graviton secure contexts \cite{Graviton:OSDI18}. 
The CPU TEE also runs Graviton's trusted GPU runtime, which enables \sys to securely bootstrap the GPU TEE and establish a single trust domain across the TEEs.
The GPU runtime talks to the untrusted GPU driver (running on the host outside the CPU TEE) to manage resources on the GPU via \code{ioctl} calls. In Graviton, each \code{ioctl} call is translated to a sequence of commands submitted to the command processor. Graviton ensures {\em secure} command submission (and subsequently \code{ioctl} delivery) as follows: 
 \begin{enumerate*}[($i$)]
 \item for task submission, the runtime uses authenticated encryption to protect commands from being dropped, replayed, or reordered, and 
 \item for resource management, the runtime validates signed summaries returned by the GPU upon completion.
 \end{enumerate*}
The GPU runtime {\em encrypts all inter-TEE communication}.

We port the non-CNN video modules (\cref{fig:pipeline}) to SGX enclaves using the Graphene LibOS~\cite{Graphene:SGX:Tsai:2017}.
In doing so, we instrument Graphene to support the \code{ioctl} calls that are used by the runtime to communicate with the GPU driver. %

\paragraph{Pipeline execution} %
The hybrid architecture requires us to protect against attacks on the CPU TEE, GPU TEE, and the CPU-GPU channel. 
As \cref{fig:architecture} illustrates, \sys decrypts the video stream inside the CPU TEE, and obliviously decodes out each frame (in \cref{s:decoding}). 
\sys then processes the decoded frames using oblivious vision algorithms to extract objects from each frame (in \cref{s:algorithms}).
\sys extracts the \emph{same} number of objects of \emph{identical dimensions} from each frame (some of which are dummies, up to an upper-bound) and feeds them into a circular buffer. 
This avoids leaking the \emph{actual} number of objects in each frame and their \emph{sizes}; the attacker can observe accesses to the buffer, even though objects are encrypted. 
Objects are dequeued from the buffer and sent to the GPU (\cref{s:system:communication}) where they are decrypted and processed obliviously by the CNN in the GPU TEE (\cref{s:algorithms:cnn}).

\subsection{CPU-GPU Communication} \label{s:system:communication}

Although the CPU-GPU channel in \cref{fig:architecture} transfers encrypted objects, \sys needs to ensure that its traffic patterns are independent of the video content.
Otherwise, an attacker observing the channel can infer the processing rate of objects, and hence the number (and size) of the detected objects in each frame. 
To address this leakage, \sys ensures that 
\begin{enumerate*}[($i$)] 
\item the CPU TEE transfers the same number of objects to the GPU per frame, and 
\item CNN inference runs at a fixed rate (or batch size) in the GPU TEE.
\end{enumerate*} %
Crucially, \sys ensures that the CNN processes as few {\em dummy objects} as possible. While our description focuses on \cref{fig:pipeline:resnet} to hide the processing rate of {\em objects of a frame} on the GPU, our techniques directly apply to the pipeline of \cref{fig:pipeline:yolo} to hide the processing rate of complete frames using {\em dummy frames}.

Since the CPU TEE already extracts a fixed number of objects per frame (say $k_\code{max})$ for obliviousness, we enforce an inference rate of $k_\code{max}$ for the CNN as well, regardless of the number of \emph{actual} objects in each frame (say $k$).
The upper bound $k_\code{max}$ is easy to learn for each video stream in practice.
However, this leads to a wastage of GPU resources, which must now also run inference on $(k_\code{max} - k)$ dummy objects per frame. 
To limit this wastage, we develop an oblivious protocol that leads to processing as few dummy objects as possible.

\paragraph{Oblivious protocol} \sys runs CNN inference on $k^\prime (<< k_\code{max})$ objects per frame.  %
\sys's CPU pipeline extracts $k_\code{max}$ objects from each frame (extracting dummy objects if needed) and pushes them into the head of the circular buffer (\cref{fig:architecture}). %
At a fixed rate (\eg once per frame, or every 33ms for a 30fps video), $k^\prime$ objects are dequeued from the \emph{tail} of the buffer and sent to the GPU that runs inference on all $k^\prime$ objects.

We reduce the number of dummy objects processed by the GPU as follows. %
We sort the buffer using \osort in ascending order of ``priority'' values (dummy objects are assigned lower priority), thus moving dummy objects to the \emph{head} of the buffer and actual objects to the \emph{tail}. Dequeuing from the tail of the buffer ensures that actual objects are processed first, and that %
dummy objects at the head of the buffer are likely {\em overwritten} %
before being sent to the GPU. 
The circular buffer's size is set large enough to avoid overwriting actual objects.%

The consumption (or inference) rate $k^\prime$ should be set relative to the actual number of objects that occur in the frames of the video stream. Too high a value of $k^\prime$ results in GPU wastage due to dummy inferences, while too low a value leads to delay in the processing of the objects in the frame (and potentially overwriting them in the circular buffer). In our experiments, we use a value of $k^\prime = 2 \times k_\code{avg}$ ($k_\code{avg}$ is the average number of objects in a frame) that leads to little delay and wastage.

\paragraph{Bandwidth consumption}
The increase in traffic on the CPU-GPU PCIe bus (\cref{fig:architecture}) due to additional dummy objects for obliviousness is not an issue because the bus is not bandwidth-constrained.
Even with \sys's oblivious video pipelines, we measure the data rate to be $<$\SI{70}{\mega\byte/\second}, in contrast to the several \SI{}{\giga\byte/\second} available in PCIe interconnects.

\subsection{CNN Classification on the GPU}\label{s:algorithms:cnn}
The CNN processes identically-sized objects at a fixed rate on the GPU.
The vast majority of CNN operations, such as matrix multiplications, have inherently input-independent access patterns~\cite{Ohrimenko:ObliviousML, Privado:ObliviousML}. 
The operations that are \emph{not} oblivious can be categorized as conditional assignments. 
For instance, the ReLU function, when given an input $x$, replaces $x$ with $max($0$, $x$)$; likewise, the max-pooling layer replaces each value within a square input array with its maximum value.

Oblivious implementation of the $max$ operator may use CUDA \code{max}/\code{fmax} intrinsics for integers/ floats, which get compiled to \code{IMNMX}/\code{FMNMX} instructions~\cite{Cuda:ISA} that execute the $max$ operation branchlessly. This ensures that the code is free of data-dependent accesses, making CNN inference oblivious. %

\subsection{Oblivious Modules on the CPU}\label{s:system:cpu}
After providing a data-oblivious CPU-GPU channel and CNN execution on the GPU, we address the video modules (in \cref{fig:pipeline}) that execute on the CPU. 
We carefully craft oblivious versions of the video modules using novel efficient algorithms (which we describe in the subsequent sections).
To implement our algorithms, we use a set of oblivious primitives which we summarize below.

\paragraph{Oblivious primitives}
\label{s:background:primitives}
We use three basic primitives, similar to prior work~\cite{Raccoon, Ohrimenko:ObliviousML, Zerotrace:Sasy}.
Fundamental to these primitives is the x86 \code{CMOV} instruction, which takes as input two registers---a source and a destination---and moves the source to the destination if a condition is true.
Once the operands have been loaded into registers, the instructions are immune to memory-access-based pattern leakage because registers are private to the processor, making any register-to-register operations oblivious by default.

\paragraph{1) Oblivious assignment \code{(\oselect)}}
The \oselect primitive is a wrapper around the \code{CMOV} instruction that conditionally assigns a value to the destination operand. This primitive can be used for performing dummy write operations by simply setting the input condition to false. We implement multiple versions of this primitive for different integer sizes. We also implement a vectorized version using SIMD instructions.%

\paragraph{2) Oblivious sort \code{(\osort)}}
The \osort primitive obliviously sorts an array with the help of a bitonic sorting network~\cite{BatcherSort}.
Given an input array of size $n$, the network sorts the array by performing $O(n \log^2(n))$ compare-and-swap operations, which can be implemented using the \oselect primitive. As the network layout is fixed given the input size $n$, execution of each network has identical memory access patterns.

\paragraph{3) Oblivious array access \code{(\oarr)}}
The \oarr primitive accesses the $i$-th element in an array, without leaking the value of $i$.
The simplest way of implementing \oarr is to scan the entire array. %
However, as discussed in our threat model (\cref{s:threatmodel:attacks}), %
hyperthreading is disabled, preventing any sharing of intra-core resources (e.g., L1 cache) with an adversary, and consequently mitigating known attacks~\cite{sgxattacks:Memjam, sgxattacks:CacheBleed} that can leak access patterns at sub-cache-line granularity using shared intra-core resources. 
Therefore, we assume access pattern leakage at the granularity of cache lines, and it suffices for \oarr to scan the array at cache-line granularity for obliviousness, instead of per element or byte.

\section{Designing Oblivious Vision Modules}
\label{s:obl_overview}

Na\"ive approaches and generic tools for oblivious execution of vision modules can lead to prohibitive performance overheads. For instance, a na\"ive approach for implementing oblivious versions of CPU video analytics modules (as in \cref{fig:pipeline}) is to simply rewrite them using the oblivious primitives outlined in \cref{s:background:primitives}.
Such an approach:
\begin{enumerate*}[(i)]
\item eliminates all branches and replaces conditional statements with \oselect operations to prevent control flow leakage via access patterns to code,
\item implements all array accesses via \oarr to prevent leakage via memory accesses to data, and
\item performs all iterations for a fixed number of times while executing dummy operations when needed.
\end{enumerate*}
The simplicity of this approach, however, comes at the cost of high overheads: two to three orders of magnitude.
Furthermore, as we show in \cref{s:evaluation:related}, generic tools for executing programs obliviously such as Raccoon~\cite{Raccoon} and Obfuscuro~\cite{Obfuscuro} also have massive overheads---six to seven orders of magnitude.

Instead, we demonstrate that by carefully crafting oblivious vision modules using the primitives outlined in \cref{s:background:primitives}, \sys improves performance over na\"ive approaches by several orders of magnitude.
In the remainder of this section, we present an overview of our design strategy, before diving into the detailed design of our algorithms in \cref{s:decoding} and \cref{s:algorithms}.

\subsection{Design Strategy}
Our overarching goal is to transform each algorithm into a pattern that processes each pixel identically, regardless of the pixel's value.
To apply this design pattern efficiently, we devise a set of algorithmic and systemic optimization strategies.
These strategies are informed by the properties of vision modules, as follows.

\paragraph{1) Divide-and-conquer for improving performance}
We break down each vision algorithm into independent subroutines based on their functionality and make each subroutine oblivious individually. %
Intuitively, this strategy improves performance by
\begin{enumerate*}[(i)]
    \item allowing us to tailor each subroutine separately, and 
    \item preventing the overheads of obliviousness from getting compounded.
\end{enumerate*}

\paragraph{2) Scan-based sequential processing}
Data-oblivious processing of images demands that each pixel in the image be indistinguishable from the others.
This requirement presents an opportunity to revisit the design of sequential image processing algorithms.
Instead of simply rewriting existing algorithms using the data-oblivious primitives from \cref{s:background:primitives},
 
we find that recasting the algorithm into a form that scans the image, while applying the same functionality to each pixel, yields superior performance.
Intuitively, this is because any non-sequential pixel access implicitly requires a scan of the image for obliviousness (\eg using \oarr); therefore, by transforming the algorithm into a scan-based algorithm, we get rid of such non-sequential accesses.

\paragraph{3) Amortize cost across groups of pixels}
Processing each pixel in an identical manner lends itself naturally to optimization strategies that enable batched computation over pixels---\eg the use of data-parallel (SIMD) instructions.

\smallskip
\noindent In \sys, we follow the general strategy above to design oblivious versions of popular vision modules that can be composed and reused across diverse pipelines. %
However, our strategy can potentially help inform the design of other oblivious vision modules as well, beyond the ones we consider.

\subsection{Input Parameters for Oblivious Algorithms}\label{s:system:parameters}
Our oblivious algorithms rely on a set of public input parameters that need to be provided to \sys before the deployment of the video pipelines.
These parameters represent various upper bounds on the properties of the video stream, such as the maximum number of objects per frame, or the maximum size of each object.

\cref{table:parameters} summarizes the list of input parameters across all the modules of the vision pipeline.

\begin{figure}
\centering
\small
\begin{tabular}[t]{p{2.9cm}|p{4.7cm}}
\thickhline
Component & Input parameters \T\B \\
			  \thickhline
			  Video decoding (\cref{s:decoding}) &
			  Number of bits used to encode each (padded) row of blocks;
			  \\\hline
			  Background sub. (\cref{s:algorithms:bgs}) & -- \\\hline
			  Bounding box detection (\cref{s:algorithms:objdet}) &  
			  \begin{enumerate*}[($i$)]
			  \item Maximum number of objects per image; \item Maximum number of different labels that can be assigned to pixels (an object consists of all labels that are adjacent to each other).
			  \end{enumerate*}\\\hline
			  Object cropping (\cref{s:algorithms:cropping}) & 
			  Upper bounds on object dimensions.
			  \\\hline
			  Object tracking (\cref{s:algorithms:tracking}) & 
			  \begin{enumerate*}[($i$)]
			  \item An upper bound on the intermediate number of features;
			  \item An upper bound on the total number of features.
			  \end{enumerate*}\\\hline
			  CNN Inference (\cref{s:algorithms:cnn}) & -- \\%\hline
				  \thickhline
				  \end{tabular}
				  \caption{Public input parameters in \sys's oblivious modules.
				  }
\label{table:parameters}
\end{figure}

There are multiple ways by which these parameters may be determined.
\begin{enumerate*}[($i$)]
\item The model owner may obtain these parameters simultaneously while training the model on a public dataset.
\item The client may perform offline empirical analysis of their video streams and choose a reasonable set of parameters.
\item \sys may also be augmented to compute these parameters dynamically, based on historical data (though we do not implement this).
\end{enumerate*}
We note that providing these parameters is not strictly necessary, but meaningful parameters can significantly improve the performance of our algorithms.

\section{Oblivious Video Decoding}\label{s:decoding}

Video encoding converts a sequence of raw images, called \emph{frames}, into a compressed bitstream. Frames are of two types: \emph{keyframes} and \emph{interframes}. Keyframes are encoded %
to only exploit redundancy across pixels {\em within the same frame}. Interframes, on the other hand, use the prior frame as reference (or the most recent keyframe), and thus can exploit temporal redundancy in pixels {\em across frames}. 

\paragraph{Encoding overview} We ground our discussion using the VP8 encoder \cite{VP8Overview}, but our techniques are broadly applicable. 
A frame is decomposed into square arrays of pixels called {\em blocks}, and then compressed using the following steps (see \cref{fig:encoding}).
\bubble{1} An estimate of the block is first \emph{predicted} using reference pixels (in a previous frame if interframe or the current frame if keyframe). The prediction is then subtracted from the actual block to obtain a \emph{residue}. 
\bubble{2} 
Each block in the residue is \emph{transformed} into the frequency domain (\eg using a discrete cosine transform), and its coefficients are \emph{quantized} %
thus improving compression. 
\iffull
At the end of this step, each block comprises a sequence of 16 data values, the last several of which are typically zeros as the quantization factors for the later coefficients are larger than those of the initial ones.
\fi
\bubble{3} Each (quantized) block is compressed into a variable-sized bitstream using a binary prefix tree and arithmetic encoding. %
\iffull
The last few coefficients that are zeros are not encoded, and an end-of-block symbol (EOB) is encoded instead.
\fi
Block prediction modes, cosine transformation, and arithmetic encoding are core to all video encoders (\eg H264~\cite{H264}, VP9~\cite{VP9}) and thus our oblivious techniques carry over to all popular codecs.

\begin{figure}[t!]
    \centering
    \includegraphics[width=0.8\linewidth]{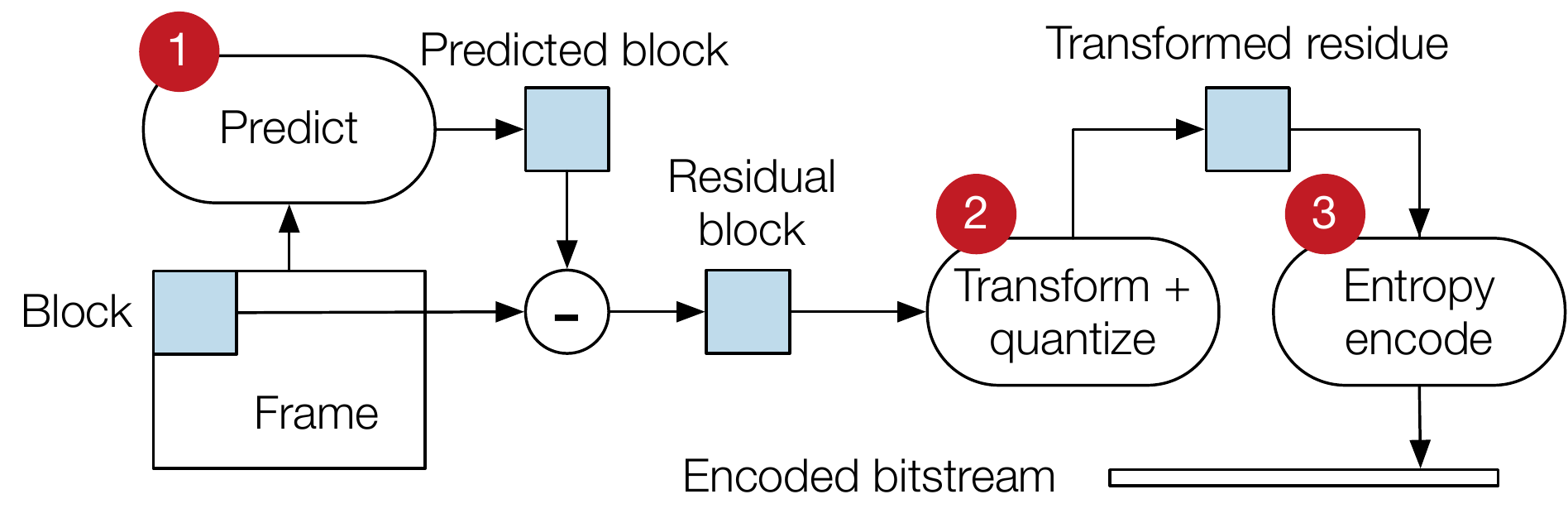}
    \caption{Flowchart of the encoding process.}
    \label{fig:encoding}
\end{figure}

The {\em decoder} reverses the steps of the encoder: 
 \begin{enumerate*}[($i$)]
 \item the incoming video bitstream is entropy decoded (\cref{s:decoding:bitstream}); 
 \item the resulting coefficients are dequantized and inverse transformed to obtain the residual block (\cref{s:decoding:transformation}); and 
 \item previously decoded pixels are used as reference to obtain a prediction block, which are then added to the residue (\cref{s:decoding:prediction}).
 \end{enumerate*}
\ifextended
\drop{
While our explanation here is simplified, 
\cref{s:proofs:decoder} provides detailed pseudocode and proofs of obliviousness for the decoder.%
}
\else
Our explanation here is simplified; we defer detailed pseudocode along with security proofs to an extended appendix~\cite{visor:extended}.%
\fi

\ifpadding
\subsection{Video Encoder Padding}\label{s:padding}
While the video stream is in transit, the bitrate variation of each frame is visible to an attacker observing the network even if the traffic is TLS-encrypted. This variability can be exploited for fingerprinting video streams~\cite{Schuster:Video:Attack} and understanding its content. %
Overcoming this leakage requires changes to the video {\em encoder} to ``pad'' each frame with dummy bits to an upper bound before sending the stream to \sys. %

We modify the video encoder to pad the encoded video streams. However, instead of applying padding at the level of frames, we pad each individual \emph{row of blocks} within the frames. Compared to frame-level padding, padding individual rows of blocks significantly improves latency of oblivious decoding, but at the cost of an increase in network bandwidth.

Padding the frames of the video stream, 
however, negates the benefit of using \emph{interframes} during encoding of the raw video stream, which are typically much smaller than keyframes. %
We therefore configure the encoder to encode all raw video frames into keyframes, which eliminates the added complexity of dealing with interframes, and consequently simplifies the oblivious decoding procedure.

We note that it may not always be possible to modify legacy cameras to incorporate padding.
In such cases, potential solutions include the deployment of a lightweight edge-compute device that pads input camera feeds before streaming them to the cloud.
For completeness, we also discuss the impact of the lack of padding in \cref{s:appendix:padding}, along with the accompanying security-performance tradeoff.
\fi

\begin{figure*}[t!]
	\centering
	\subfigure[\bf CCL-based algorithm for bounding box detection]{
		\includegraphics[width=0.46\textwidth]{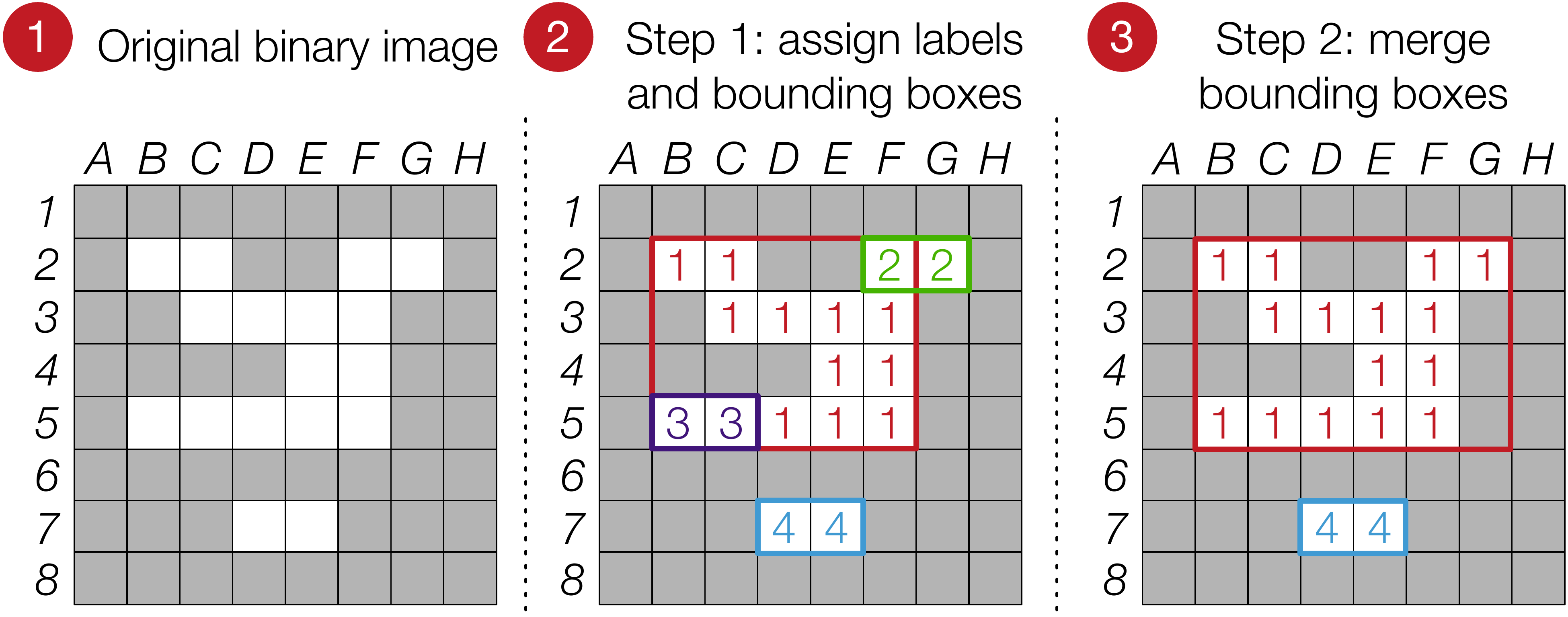}
	\label{fig:ccl}}
	\hspace{.25in}
	\subfigure[\bf Enhancement via parallelization]{
		\includegraphics[width=0.46\textwidth]{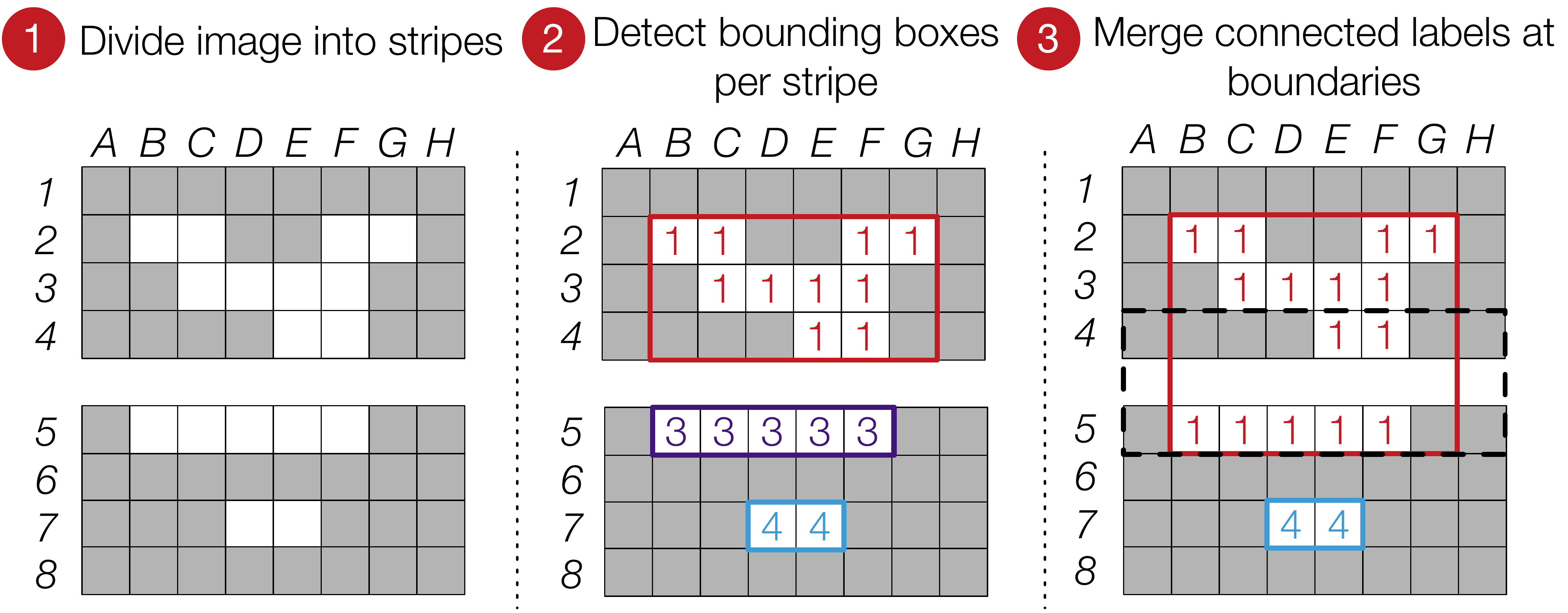}
	\label{fig:ccl:parallel}}
	\caption{Oblivious bounding box detection}
\end{figure*}

\subsection{Bitstream Decoding}
\label{s:decoding:bitstream}

The bitstream decoder reconstructs blocks with the help of a \emph{prefix tree}.
At each node in the tree it decodes a single bit from the compressed bitstream via arithmetic decoding, and traverses the tree based on the value of the bit. 
While decoding the bit, the decoder first checks whether any more bits can be decoded at the current bitstream position, and if not, it advances the bitstream pointer by two bytes.
Once it reaches a leaf node, it outputs a coefficient based on the position of the leaf, and assigns the coefficient to the current pixel in the block. %
\iffull
If an EOB symbol is decoded, then all the coefficients remaining in the block are assigned a value of zero.
\fi
This continues for all the coefficients in the frame.%

\paragraph{Requirements for obliviousness}
The above algorithm leaks information about the compressed bitstream. First, the traversal of the tree leaks the \emph{value of the parsed coefficient}.
For obliviousness, we need to ensure that during traversal, the identity of the current node being processed remains secret. 
Second, not every position in the bitstream encodes the same number of coefficients, and the bitstream pointer advances variably during decoding. Hence, this leaks the \emph{number of coefficients} that are encoded per two-byte chunk (which may convey their values).
\iffull
Finally, the presence of EOB coefficients, coupled with the assignment of decoded coefficients to pixels, leaks the number of zero coefficients per block of the frame---prior work has demonstrated attacks that exploit similar leakage to infer the outlines of all objects in the frame~\cite{sgxattacks-xu:pagefaults}.
\fi
We design a solution that \emph{decouples} the parsing of coefficients, \ie prefix tree traversal (\cref{s:decoding:traversal}), from the assignment of the parsed coefficients to pixels (\cref{s:decoding:assignment}).

\subsubsection{Oblivious prefix tree traversal}\label{s:decoding:traversal}
A simple way to make tree traversal oblivious is to represent the prefix tree as an array.
We can then obliviously fetch any node in the tree using \oarr (\cref{s:background:primitives}).
Though this hides the identity of the fetched node, we need to also ensure that \emph{processing} of the nodes does not leak their identity.

In particular, we need to ensure that nodes are indistinguishable from each other by performing an identical set of operations at each node.
Unfortunately, this requirement is complicated by the following facts. (1)~Only leaf nodes in the tree produce outputs (\ie the parsed coefficients) and not the intermediate nodes. (2)~We do not know beforehand which nodes in the tree will cause the bitstream pointer to be advanced; at the same time, we need to ensure that the pointer is advanced predictably and independent of the bitstream.
To solve these problems, we take the following steps.
\begin{compactenumerate}
\item We modify each node to output a coefficient regardless of whether it is a leaf state or not. Leaves output the parsed coefficient, while other states output a dummy value.
\item We introduce a dummy node into the prefix tree. %
While traversing the tree, if no more bits can be decoded at the current bitstream position, we transition to the dummy node and perform a bounded number of dummy decodes.%
\end{compactenumerate}
These modifications ensure that while traversing the prefix tree, all that an attacker sees is that at \emph{some} node in the tree, a single bit was decoded and a single value was outputted. %

Note that in this phase, we do not assign coefficients to pixels, and instead collect them in a list. 
If we were to assign coefficients to pixels in this phase, then the decoder would need to obliviously scan the entire frame (using \oarr) at every node in the tree, in order to hide the pixel's identity. %
Instead, by \emph{decoupling} parsing from assignment, we are able to perform the assignment obliviously using a super-linear number of accesses (instead of quadratic), as we explain next.

\iffull
\begin{figure}
    \centering
    \includegraphics[width=\linewidth]{figs/coeffsort.pdf}
    \caption{Steps for obliviously sorting the coefficients into place after populating it with zero coefficients. For simplicity, this illustration assumes that there are two subblocks, with three coefficients per subblock.}
    \label{fig:coeffs}
\end{figure}
\fi

\subsubsection{Oblivious coefficient assignment}\label{s:decoding:assignment} 
At the end of \cref{s:decoding:traversal}, we have a list of actual and dummy coefficients. %
The key idea is that if we can obliviously sort this set of values using \osort such that  all the actual coefficients are contiguously ordered while all dummies are pushed to the front,
then we can simply read the coefficients off the end of the list sequentially and assign them to pixels one by one. %

\iffull
    However, recall that in lieu of the trailing zeros within each block, the encoder encodes an EOB symbol instead. Therefore, we need to append the requisite zeros to the set and move them to the appropriate indices before we can carry out the assignment.
    To achieve this, our algorithm makes a single forward pass over the set to add the zeros, while updating all index values per tuple in a way that ensures the zeros will be sorted to the correct positions, as illustrated in \cref{fig:coeffs}.
\fi

To enable such a sort, we modify the prefix tree traversal to additionally output a tuple $(\code{flag}, \code{index})$ per coefficient; \code{flag} is  0 for dummies and 1 otherwise; \code{index} is an increasing counter as per the pixel's index.
Then, the desired sort can be achieved by sorting the list based on the value of the tuple.

\ifpadding
As the complexity of oblivious sort is super-linear in the number of elements being sorted,
an important optimization is to decode and assign coefficients to pixels at the granularity of {\em rows of blocks} rather than frames.
While the number of bits per row of blocks may be observed, the algorithm's obliviousness is not affected as each row of blocks in the video stream is padded to an upper bound (\cref{s:padding}); had we applied frame-level padding, this optimization would have revealed the number of bits per row of blocks.
In \cref{s:eval:decoding}, we show that this technique improves oblivious decoding latency by \approx$6\times$.
\fi

\subsection{Dequantization and Inverse Transformation} \label{s:decoding:transformation}
The next step in the decoding process is to 
 \begin{enumerate*}[($i$)]
 \item dequantize the coefficients decoded from the bitstream, followed by 
 \item inverse transformation to obtain the residual blocks.
 \end{enumerate*}
Dequantization just multiplies each coefficient by a quantization factor. 
The inverse transformation also %
performs a set of identical arithmetic operations irrespective of the coefficient values.

\subsection{Block Prediction}\label{s:decoding:prediction}

Prediction is the final stage in decoding. The residual block obtained after \cref{s:decoding:transformation} is added to a {\em predicted block}, obtained using a previously constructed block as reference, to obtain the raw pixel values. 
In keyframes, each block is \emph{intra}-predicted---\ie it uses a block in the same frame as referenced.
\ifpadding
We do not discuss interframes because as described in \cref{s:padding}, the padded input video streams in \sys only contain keyframes.
\fi

Intra-predicted blocks are computed using one of several \emph{modes}. %
A mode to encode a block refers to a combination of pixels on its top row and left column used as reference. %
Obliviousness requires that the prediction mode %
remains private.
Otherwise, an attacker can identify the pixels that are most similar to each other, thus revealing details about the frame.

\iffull
We investigate two different ideas for making intra-prediction oblivious. 
First, we note that in each prediction mode, the value of a pixel in the predicted block can be expressed as a linear combination $\Sigma a_i p_i$ of all the pixels that lie above and to the left of the block.
Here $p_i$ represents the adjoining pixels and $a_i$ are weights. 
Thus, to compute the value of the predicted pixel obliviously, we can simply evaluate the expression after using the \oselect primitive to obliviously assign each $a_i$ a value based on the mode and the location of the current pixel. 

A second approach is to simply evaluate all possible predictions for the pixel and store them in an array, indexing each prediction by its mode.
Then, use the \oarr primitive to obliviously select the correct prediction from the array.

We implemented both approaches, and found that the second offers better performance in practice.
This is because in the second approach, we can compute the predicted values for several pixels simultaneously at the level of individual rows, which amortizes the cost of our operations. 
\fi

We make intra-prediction oblivious by evaluating all possible predictions for the pixel and storing them in an array, indexing each prediction by its mode. Then, we use \oarr to obliviously select the correct prediction from the array. %

\section{Oblivious Image Processing}\label{s:algorithms}

\begin{figure*}[t!]
    \centering
    \subfigure[\bf Localizing objects.]{
        \centering
        \includegraphics[width=0.28\textwidth]{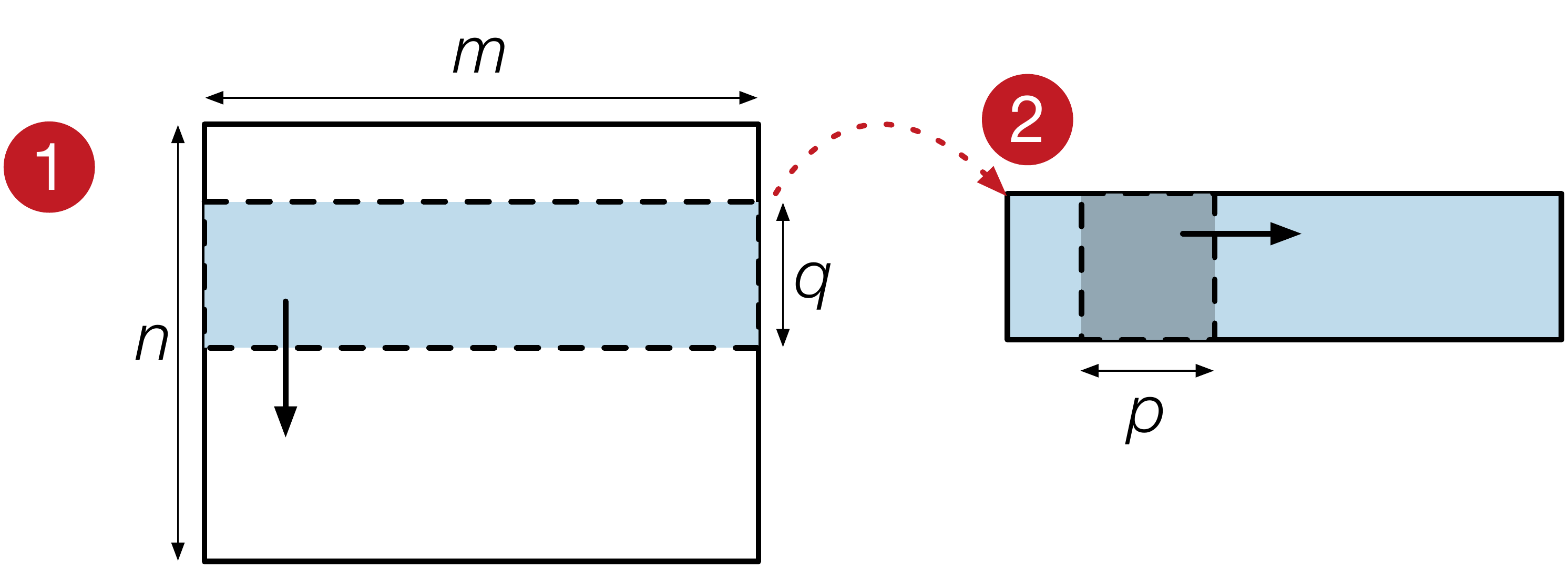}
    \label{fig:cropv1}}
    \subfigure[\bf Bilinear interpolation.]{
        \centering
        \includegraphics[width=0.29\textwidth]{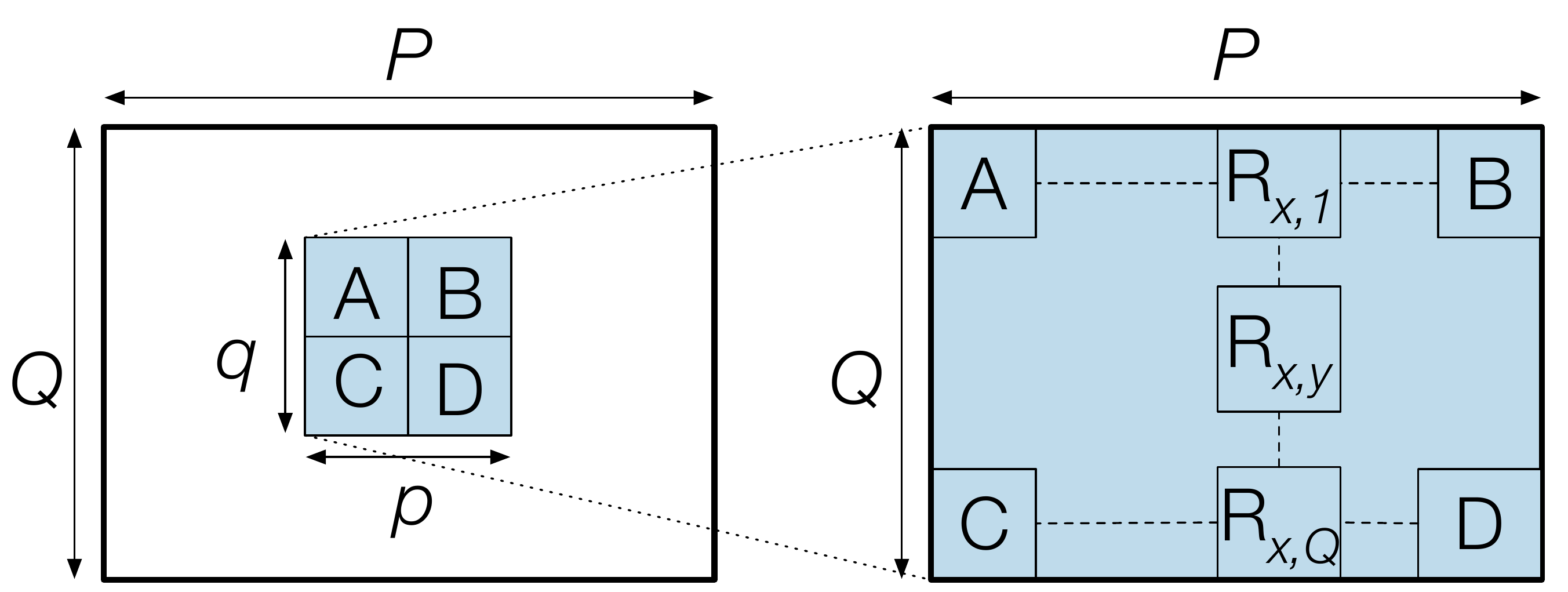}
        \label{fig:resize:v1}
    }
    \subfigure[\bf Improved Bilinear interpolation.]{
        \centering
        \includegraphics[width=0.34\textwidth]{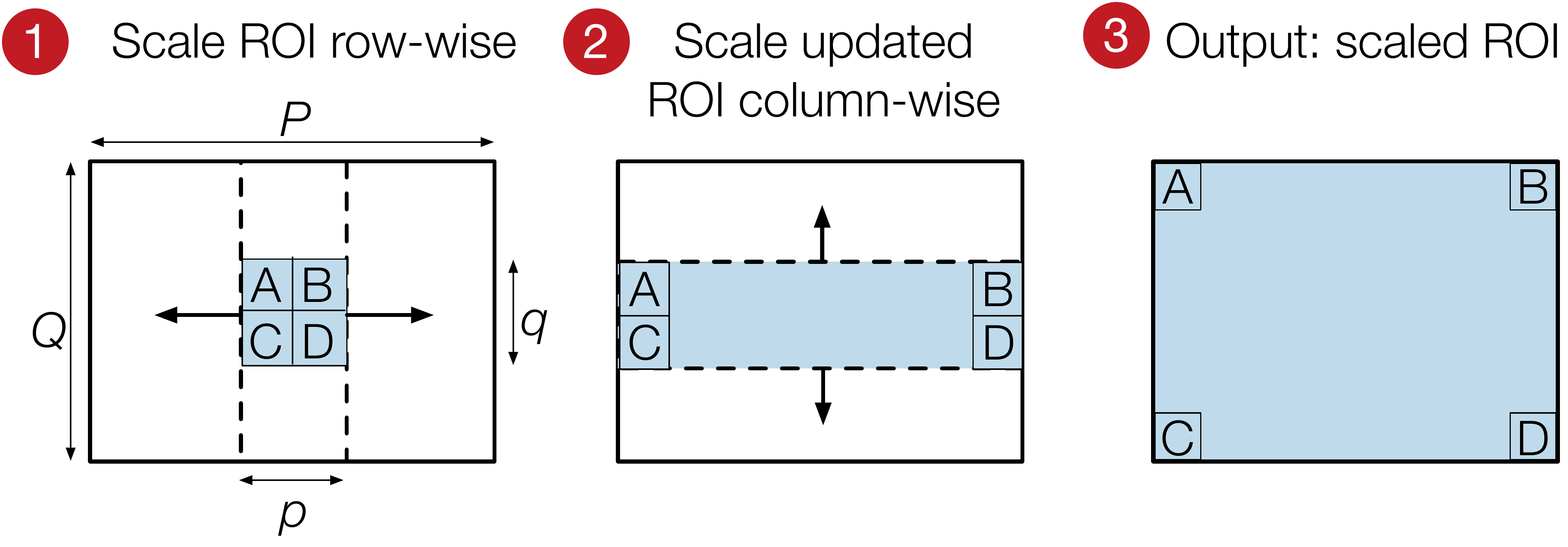}
        \label{fig:resize:v2}
    }
    \caption{Oblivious object cropping}
\end{figure*}

After obliviously decoding frames in \cref{s:decoding}, the next step as shown in \cref{fig:pipeline} is to develop data-oblivious techniques for background subtraction (\cref{s:algorithms:bgs}), bounding box detection (\cref{s:algorithms:objdet}), object cropping (\cref{s:algorithms:cropping}), and tracking (\cref{s:algorithms:tracking}).
We present the key ideas here; detailed pseudocode and proofs of obliviousness are available in 
\ifextended
\cref{s:proofs:algorithms}. 
\else
an extended appendix~\cite{visor:extended}.
\fi
Note that \cref{s:algorithms:bgs} and \cref{s:algorithms:tracking} modify popular algorithms to make them oblivious, while \cref{s:algorithms:objdet} and \cref{s:algorithms:cropping} propose new oblivious algorithms.

\subsection{Background Subtraction}\label{s:algorithms:bgs}
The goal of background subtraction is to detect moving objects in a video.
Specifically, it dynamically learns stationary pixels that belong to the video's background, and then subtracts them from each frame, %
thus producing a binary image with black background pixels and white foreground pixels.

Zivkovic \etal proposed a mechanism~\cite{MOG2:Zivkovic:2004,MOG2:Zivkovic:2006} that is widely used in practical deployments, that models each pixel %
as a mixture of Gaussians~\cite{MOG:Survey:2008}. The number of Gaussian components $M$ differs across pixels depending on their value (but is no more than $M_\code{max}$, a pre-defined constant). As more data arrives (with new frames), the algorithm updates each Gaussian component along with their weights ($\pi$), and adds new components if necessary.

To determine if a pixel $\vec{x}$ belongs to the background or not, the algorithm uses the $B$ Gaussian components with the largest weights and outputs true if $p(\vec{x})$ is larger than a threshold:
{
\setlength{\abovedisplayskip}{0pt}
\setlength{\belowdisplayskip}{0pt}
\begin{align*}
p(\vec{x}) = \sum^B_{m = 1}\pi_m\mathcal{N}(\vec{x}~|~\vec{\mu}_m,\Sigma_m)
\end{align*}
}
where
$\vec{\mu}_m$ and $\Sigma_m$ are parameters of the Gaussian components, and $\pi_m$ is the weight of the $m$-th Gaussian component.

This algorithm is not oblivious because it maintains a different number of Gaussian components per pixel, and thus performs different steps while updating the mixture model per pixel.
These differences are visible via access patterns, and these leakages reveal to an attacker how \emph{complex} a pixel is in relation to others---\ie whether a pixel's value stays stable over time or changes frequently. 
This enables the attacker to identify the positions of moving objects in the video.

For obliviousness, we need to perform an identical set of operations per pixel (regardless of their value); we thus 
{\em always} maintain $M_\code{max}$ Gaussian components for each pixel, of which $(M_\code{max} - M)$ are dummy components and assigned a weight $\pi=0$.
When newer frames arrive, we use \oselect operations to make all the updates to the mixture model, making dummy operations for the dummy components.
Similarly, to select the $B$ largest components by weight, %
we use the \osort primitive.%

\subsection{Bounding Box Detection}\label{s:algorithms:objdet}

The output from \S\ref{s:algorithms:bgs} is a binary image with black background pixels where the foreground objects are white blobs (\cref{fig:ccl}).
To find these objects, it suffices to find the {\em edge contours} of all blobs.
These are used to compute the \emph{bounding rectangular box} of each object. %
A standard approach for finding the contours in a binary image
is the border following algorithm of Suzuki and Abe~\cite{Suzuki:1985}. %
As the name suggests, the algorithm works by scanning the image until it locates an edge pixel, and then follows the edge around a blob. %
As \cref{fig:leakage} in \cref{s:background:attacks} illustrated, the memory access patterns of this algorithm leak the details of all the objects in the frame.%

A na\"ive way to make this algorithm oblivious is to implement each pixel access using the \oarr primitive (along with other minor modifications). However, we measure that this approach slows down the algorithm by over \approx$1200\times$.%

We devise a two-pass oblivious algorithm for computing bounding boxes by adapting the classical technique of connected component labeling (CCL)~\cite{Rosenfeld:1966}. The algorithm’s main steps are illustrated in \cref{fig:ccl} (whose original binary image contains two blobs). In the first pass, it scans the image and assigns each pixel a temporary label if it is ``connected'' to other pixels. %
In the second pass, it merges labels that are part of a single object. 
Even though CCL on its own is less efficient for detecting blobs than border following, it is far more amenable to being adapted for obliviousness.

We make this algorithm oblivious as follows.
First, we perform identical operations regardless of whether the current pixel is connected to other pixels. %
Second, for efficiency, we restrict the maximum number of temporary labels (in the first pass) to a parameter $N$ provided as input to \sys (per \cref{s:system:parameters}, \cref{table:parameters}). %
Note that the value of the parameter may be much lower than the worst case upper bound (which is the total number of pixels), and thus is more efficient.

\paragraph{Enhancement via parallelization}
We observe that the oblivious algorithm can be parallelized using a divide-and-conquer approach. %
We divide the frame into horizontal \emph{stripes} (\bubble{1} in \cref{fig:ccl:parallel}) and process {\em each stripe in parallel} (\bubble{2}). 
For objects that span stripe boundaries, each stripe outputs only a \emph{partial} bounding box containing the pixels within the stripe. We combine the partial boxes by re-applying the oblivious CCL algorithm to the boundaries of adjacent stripes (\bubble{3}). 
Given two adjacent stripes $S_{i}$ and $S_{i+1}$ one below the other, we compare each pixel in the top row of $S_{i+1}$ with its neighbors in the bottom row of $S_{i}$, and merge their labels as required.

\subsection{Object Cropping}\label{s:algorithms:cropping}
The next step after detecting bounding boxes of objects is to {\em crop} them out of the frame to be sent for CNN classification (\cref{fig:pipeline:resnet}). %
\sys needs to ensure that the cropping of objects does not leak \begin{enumerate*}[($i$)] \item their positions, or \item their dimensions.\end{enumerate*}

\subsubsection{Hiding object positions}\label{s:algorithms:cropping:positions}

A na\"ive way of obliviously cropping an object of size $p\times q$ is to slide a window (of size $p\times q$) horizontally %
in raster order, and copy the window's pixels %
if it aligns with the object's bounding box. Otherwise, perform a dummy copy. This, however, leads to a slow down of $4000\times$, with the major reason being redundant copies: while sliding the window forward by one pixel results in a new position in the frame, a majority of the pixels copied are the same as in the previous position. 
\iffull
For a $m\times n$ frame, and an object of size $p\times q$, the technique results in $pq(m-p)(n-q)$ pixel copies, as compared to $pq$ pixel copies when directly cropping the object.
\fi

We get rid of this redundancy by \emph{decoupling} the algorithm into multiple passes---one pass along each dimension of the image---such that each pass performs only a subset of the work. %
As \cref{fig:cropv1} shows, the first phase extracts the horizontal strip containing the object; the second phase extracts the object from the horizontal strip.

\bubble{1} Instead of sliding a window (of size $p \times q$) across the frame (of size $m \times n$), we use a horizontal strip of $m \times q$ that has width $m$ equal to that of the frame, and height $q$ equal to that of the object. 
We slide the strip vertically down the frame {\em row by row}. 
If the top and bottom edges of the strip are aligned with the object, we copy all pixels covered by the strip into the buffer; otherwise, we perform dummy copies. 
\iffull
This phase results in $mq(n-q)$ pixel copies.
\fi

\bubble{2} We allocate a window of size $p\times q$ equal to the object's size and then slide it {\em column by column} across the extracted strip in \bubble{1}. If the left and right edges of the window are aligned with the object's bounding box, we copy the window's pixels into the buffer; if not, we perform dummy copies.
\iffull
This phase performs $pq(m-p)$ pixel copies.
\fi

\ifextended
\cref{alg:crop} in \cref{s:proofs:algorithms} provides the detailed steps.
\fi

\subsubsection{Hiding object dimensions}\label{s:algorithms:cropping:size}
The algorithm in \cref{s:algorithms:cropping:positions} leaks the dimensions $p\times q$ of the objects.
To hide object dimensions, \sys takes as input parameters $P$ and $Q$ representing upper bounds on object dimensions (as described in \cref{s:system:parameters}, \cref{table:parameters}), 
and instead of cropping out the exact $p \times q$ object, %
we obliviously crop out a larger image of size $P \times Q$ that \emph{subsumes} the object. While the object sizes vary depending on their position in the frame (e.g., near or far from the camera), the maximum values ($P$ and $Q$) can be learned from profiling just a few sample minutes of the video, and they tend to remain unchanged in our datasets. 

This larger image now contains extraneous pixels surrounding the object, which might lead to errors during the CNN's object classification.
We remove the extraneous pixels surrounding the $p \times q$ object by obliviously {scaling it up} to fill the $P \times Q$ buffer. Note that all objects we send to the CNN across the CPU-GPU channel are of size $P \times Q$ (\cref{s:system:communication}), and recall from \cref{s:system:architecture} that we extract the same number of objects from each frame (by padding dummy objects, if needed).

We develop an oblivious routine for scaling up using bilinear interpolation%
~\cite{Fundamentals:Jain:1989}. 
 
Bilinear interpolation computes the value of a pixel in the scaled up image using a linear combination of a $2\times2$ array of pixels from the original image (see \cref{fig:resize:v1}). 
\iffull
The simplest way to implement this routine obliviously is to fetch the $4$ pixel values obliviously using \oarr for each pixel in the scaled up image. 
This would entail $PQ$ scans of the entire image,yielding a total of $O(P^2Q^2)$ pixel accesses.
\fi
We once again use decoupling of the algorithm into two passes to improve its efficiency (\cref{fig:resize:v2}) by scaling up along a single dimension per pass.
\iffull
The two passes perform a total of $O(P^2 Q + PQ^2)$ pixel accesses, improving  asymptotic performance over the $O(P^2Q^2)$ algorithm.
\fi

\paragraph{Cache locality}
Since the second pass of our (decoupled bilinear interpolation) algorithm performs column-wise interpolations, each pixel access during the interpolation touches a different cache line.
To exploit cache locality, we \emph{transpose} the image before the second pass, and make the second pass to also perform {\em row-wise} interpolations (as in the first pass).
This results in another order of magnitude speedup (\cref{s:eval:cropping}).

\subsection{Object Tracking}\label{s:algorithms:tracking}

Object tracking consists of two main steps: feature detection in {\em each frame} and feature matching {\em across frames}. 

\paragraph{Feature detection} SIFT \cite{SIFT:Lowe:1999,SIFT:Lowe:2004} is a popular algorithm for extracting features for {\em keypoints}, i.e., pixels that are the most ``valuable'' in the frame. In a nutshell, it generates candidate keypoints, where each candidate is a local maxima/minima; the candidates are then filtered to get the legitimate keypoints.%

Based on the access patterns of the SIFT algorithm, an attacker can infer the locations of all the keypoints in the image, which in turn, can reveal the location of all object ``corners'' in the image. A na\"ive way of making the algorithm oblivious is to treat each pixel as a keypoint, performing all the above operations for each. 
However, the SIFT algorithm's performance depends critically on its ability to filter out a small set of good keypoints from the frame.

To be oblivious {\em and} efficient, \sys takes as input two parameters $N_\code{temp}$ and $N$ (per \cref{table:parameters}).
    The parameter $N_\code{temp}$ represents an upper bound on the number of candidate keypoints, and $N$ on the number of legitimate keypoints. 
These parameters, coupled with \oselect and \osort, allow for efficient and oblivious identification of keypoints.
Finally, computing the {\em feature descriptors} for each keypoint requires accessing the pixels around it. For this, we use oblivious extraction (\cref{s:algorithms:cropping}). 
\ifextended
\cref{s:proofs:algorithms}'s Algorithm \ref{alg:featuredet} has the pseudocode.%
\fi

\paragraph{Feature matching}
The next step after detecting features is to match them across images.
Feature matching computes a distance metric between two sets of features, and identifies features that are ``nearest'' to each other in the two sets.
In \sys, we simply perform brute-force matching of the two sets, using \oselect operations to select the closest features.%

\section{Evaluation}\label{s:evaluation}

\ifpadding
    \begin{figure*}[t!]
        \minipage[t]{0.32\textwidth}
        \centering
        \includegraphics[width=\linewidth]{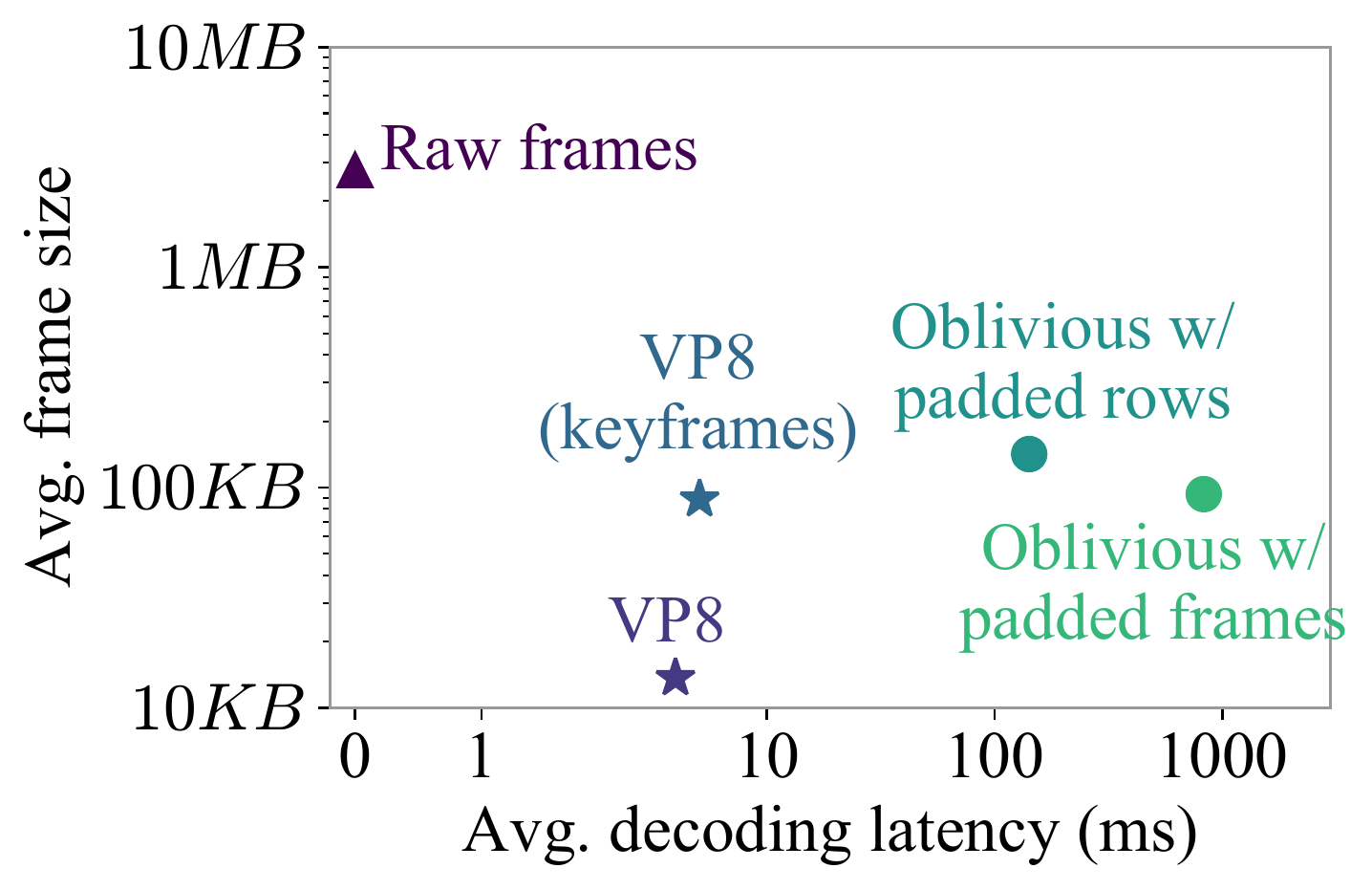}
        \caption{Decoding latency vs.\ B/W.}
        \label{fig:frames:lat}
        \endminipage\hfill
        \minipage[t]{0.32\textwidth}
        \centering
        \includegraphics[width=\linewidth]{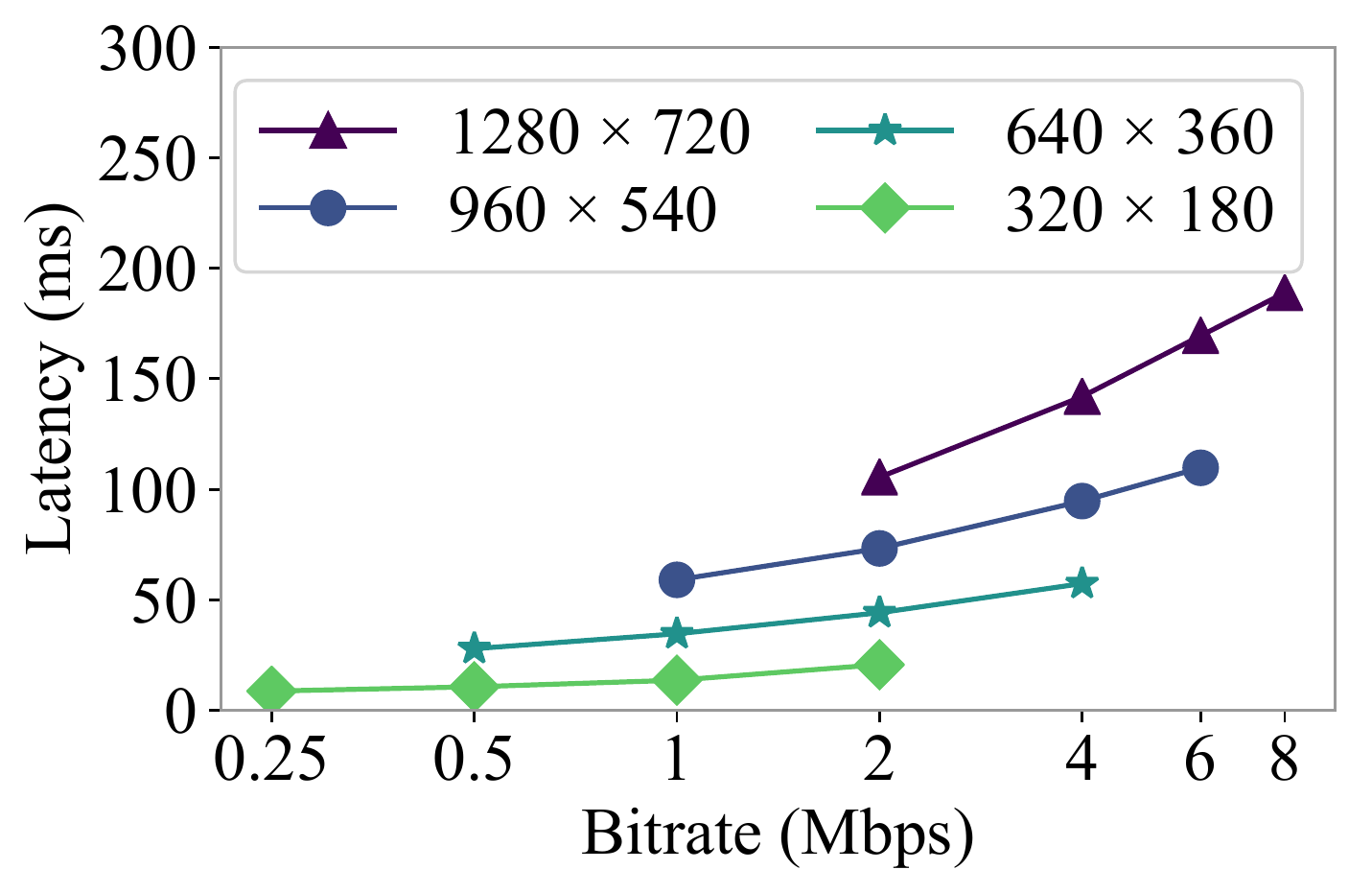}
        \caption{Latency of oblivious decoding.
        }
        \label{fig:decode:lat}
        \endminipage\hfill
        \minipage[t]{0.32\textwidth}%
        \centering
        \includegraphics[width=\linewidth]{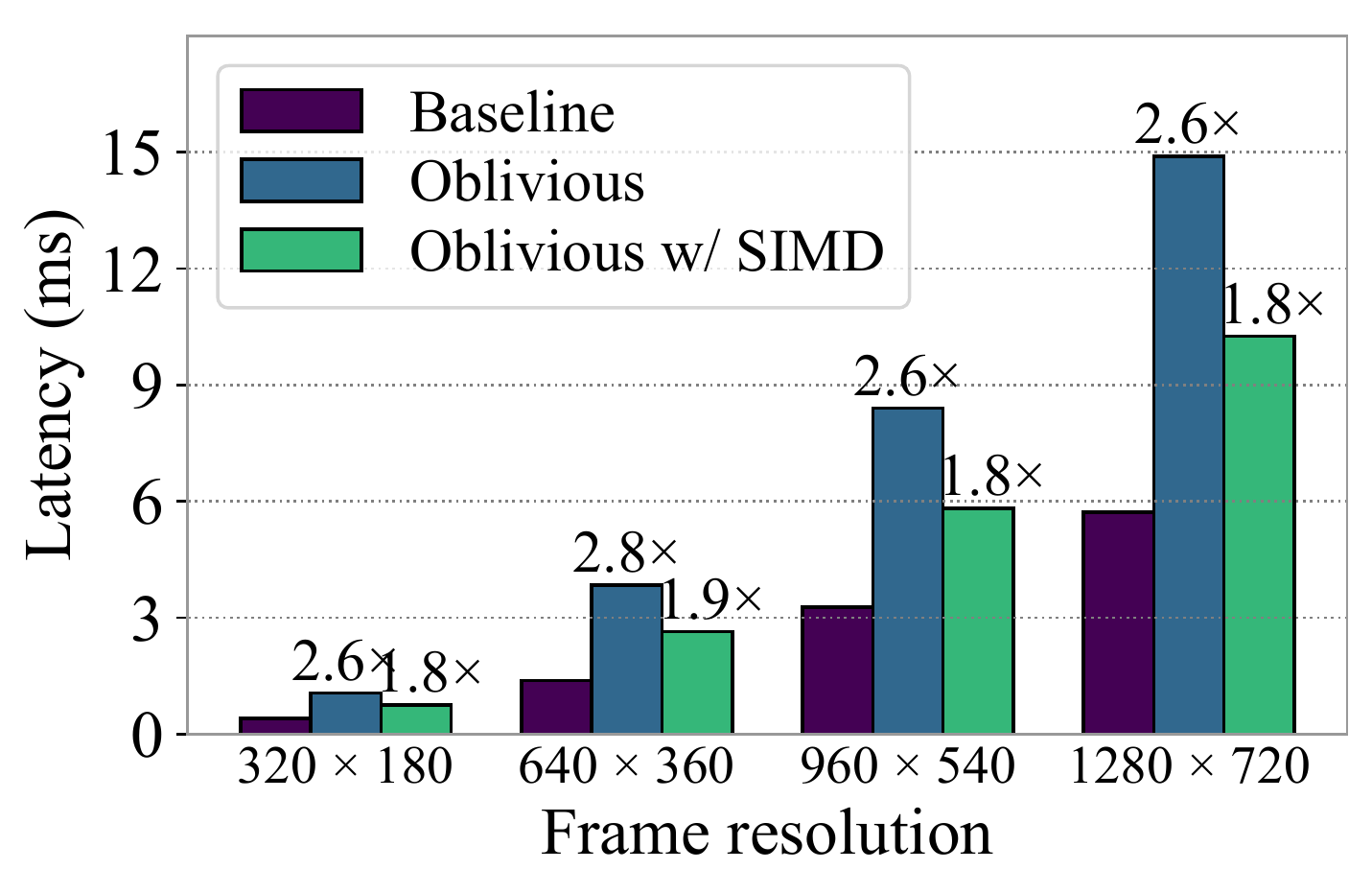}
        \caption{Background subtraction.}
        \label{fig:bgs:latency}
        \endminipage
    \end{figure*}
\else
\begin{figure*}[t!]
\minipage[t]{0.32\textwidth}
    \centering
    \includegraphics[width=\linewidth]{plots/frames_no_padding.pdf}
    \caption{Decoding latency vs.\ B/W.}
    \label{fig:frames:lat}
\endminipage\hfill
\minipage[t]{0.32\textwidth}
    \centering
    \includegraphics[width=\linewidth]{plots/decode_no_padding.pdf}
    \caption{Latency of oblivious decoding.
    }
    \label{fig:decode:lat}
\endminipage\hfill
\minipage[t]{0.32\textwidth}%
    \centering
    \includegraphics[width=\linewidth]{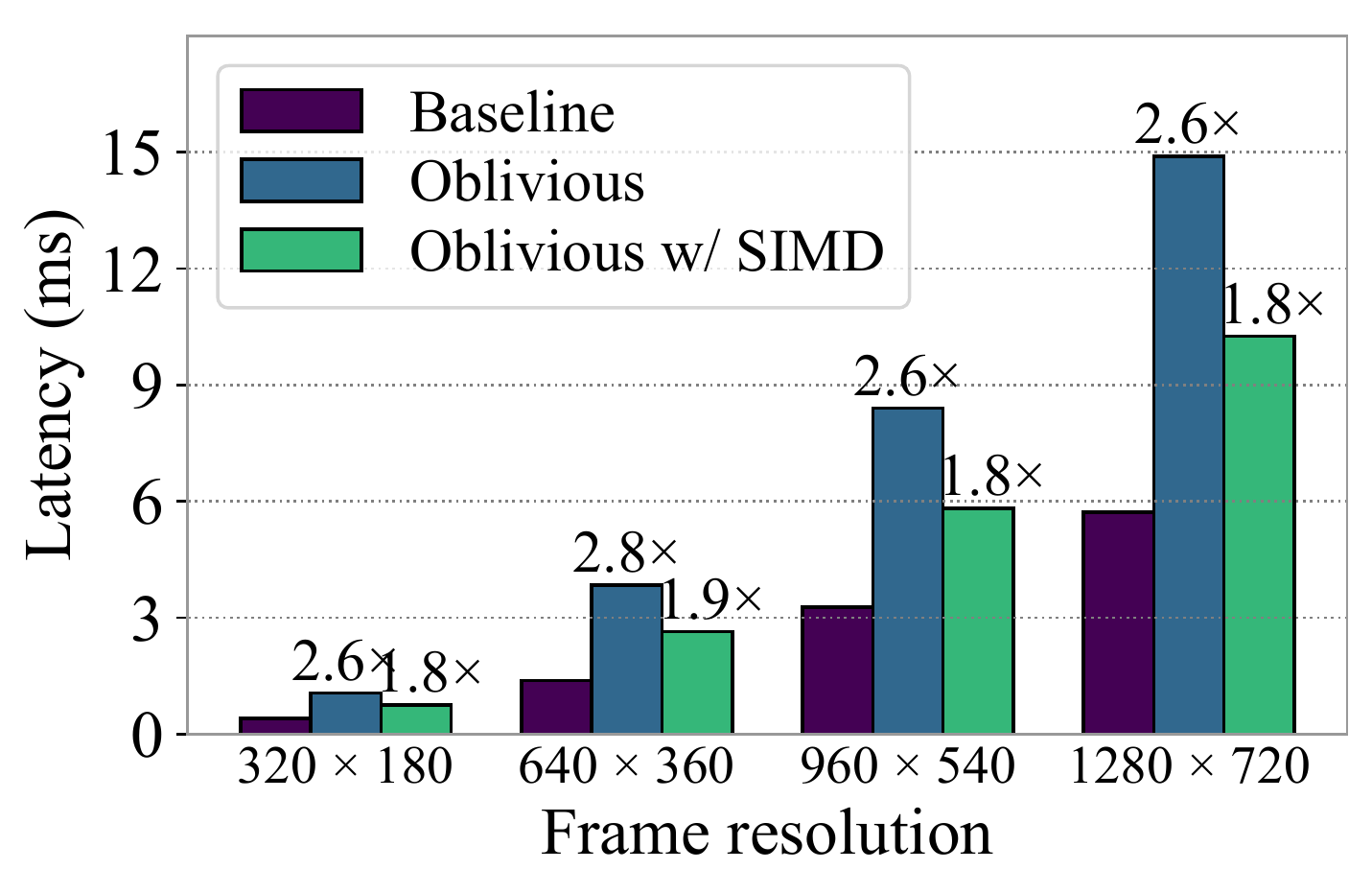}
    \caption{Background subtraction.}
    \label{fig:bgs:latency}
\endminipage
\end{figure*}
\fi

\paragraph{Implementation}
We implement our oblivious video decoder atop FFmpeg's VP8 decoder~\cite{ffmpeg} and oblivious vision algorithms atop OpenCV 3.2.0~\cite{opencv}. We use Caffe~\cite{Jia:Caffe:2014} for running CNNs. We encrypt data channels using AES-GCM. 
We implement the oblivious primitives of \cref{s:background:primitives} using inline assembly code (as in~\cite{Ohrimenko:ObliviousML,Raccoon,Zerotrace:Sasy}), and manually verified the binary to ensure that compiler optimizations do not undo our intent; one can also use tools such as Vale~\cite{Vale:Bond:2017} to do the same.

\paragraph{Testbed} We evaluate \sys on Intel i7-8700K with 6 cores running at \SI{3.7}{\giga\hertz}, and an NVIDIA GTX 780 GPU with 2304 CUDA cores running at \SI{863}{\mega\hertz}.
We disable hyperthreading for experiments with \sys (per \cref{s:threatmodel}), but retain hyperthreading in the insecure baseline.
Disabling hyperthreading for security does not sacrifice the performance of Visor (due to its heavy utilization of vector units) unlike the baseline system that favors hyperthreading; see \cref{s:hyperthreading} for more details.
The server runs Linux v4.11; supports AVX2 and SGX-v1 instruction sets; %
and has \SI{32}{\giga\byte} of memory, with \SI{93.5}{\mega\byte} of enclave memory. The GPU has \SI{3}{\giga\byte} of memory. 

\paragraph{Datasets}
We use four real-world video streams (obtained with permission) in our experiments: streams 1 and 4 are from traffic cameras in the city of Bellevue (resolution $1280 \times 720$) while streams 2 and 3 are sourced from cameras surveilling commercial datacenters (resolution $1024 \times 768$). 
All these videos are privacy-sensitive as they involve government regulations or business sensitivity. For experiments that evaluate the cost of obliviousness across different resolutions and bitrates, we re-encode the videos accordingly.
A recent body of work~\cite{NoScope,Chameleon:Sigcomm,VideoStorm} has found that the accuracy of object detection in video streams is not affected if the resolution is decreased (while consuming significantly lesser resources), and 720p videos suffice. We therefore chose to use streams closer to 720p in resolution because we believe they would be a more accurate representation of real performance.

\paragraph{Evaluation highlights} We summarize the key takeaways of our evaluation.
\begin{compactenumerate}
    \item \sys's optimized oblivious algorithms (\cref{s:decoding}, \cref{s:algorithms}) are up to $1000\times$ faster than na\"ive competing solutions. (\cref{s:eval:modules})
    \item End-to-end overhead of obliviousness for real-world video pipelines with state-of-the-art CNNs are limited to \finaloverhead over a {\em non-oblivious} baseline. (\cref{s:eval:overall})
    \item \sys is generic and can accommodate multiple pipelines (\cref{s:model}; \cref{fig:pipeline}) that combine the different vision processing algorithms and CNNs. (\cref{s:eval:overall})
	\item \sys's performance is over 6 to 7 orders of magnitude better than a state-of-the-art general-purpose system for oblivious program execution. (\cref{s:evaluation:related})
\end{compactenumerate}
Overall, \sys's use of properties of the video streams has {\em no impact on the accuracy} of the analytics outputs.

\subsection{Performance of Oblivious Components}\label{s:eval:modules}
We begin by studying the performance of \sys's oblivious modules:
we quantify the raw overhead of our algorithms (without enclaves) over non-oblivious baselines;
we also measure the improvements over na\"ive oblivious solutions.

\subsubsection{Oblivious video decoding}\label{s:eval:decoding}

\ifpadding
Decoding of the compressed bitstream dominates decoding latency, consuming up to \approx$90\%$ of the total latency. %
Further, this stage is dominated by the oblivious assignment subroutine which sorts coefficients into the correct pixel positions using \osort, consuming up to \approx$83\%$ of the decoding latency. %
Since the complexity of oblivious sort is super-linear in the number of elements being sorted,
our technique for decoding at the granularity of \emph{rows of blocks} rather than frames significantly improves the latency of oblivious decoding.

\paragraph{Overheads} \cref{fig:frames:lat} shows
the bandwidth usage
and decoding latency for different oblivious decoding strategies (\ie decoding at the level of frames, or at the level of \emph{row of blocks}) for a video stream of resolution $1280 \times 720$. We also include two reference points: non-encoded frames and VP8 encoding. The baseline latency of decoding VP8 encoded frames is $4$--\SI{5}{\milli\second}.
Non-encoded raw frames incur no decoding latency but result in frames that are three orders of magnitude larger than the VP8 average frame size (10s of \SI{}{\kilo\byte}) at  a bitrate of \SI{4}{\mega\bps}.

Frame-level oblivious decoding introduces high latency (\approx\SI{850}{\milli\second}), which is two orders of magnitude higher than non-oblivious counterparts.
Furthermore, padding each frame to prevent leakage of the frame's bitrate increases the average frame size to  \approx\SI{95}{\kilo\byte}. 
On the contrary, oblivious decoding at the level of rows of blocks delivers \approx\SI{140}{\milli\second}, which is \approx$6\times$ lower than frame-level decoding. 
However, this comes with a modest increase in network bandwidth as the encoder needs to pad each row of blocks individually, rather than a frame. In particular, the frame size increases from \approx\SI{95}{\kilo\byte} to \approx\SI{140}{\kilo\byte}.

Apart from the granularity of decoding, the latency of the oblivious sort is also governed by:  \begin{enumerate*}[($i$)] \item the frame's resolution, and \item the bitrate.\end{enumerate*}
The higher the frame's resolution / bitrate, the more coefficients there are to be sorted. 
\cref{fig:decode:lat} plots oblivious decoding latency at the granularity of rows of blocks across video streams with different resolutions and bitrates. The figure shows that lower resolution/bitrates introduce lower decoding overheads. %
In many cases, lower image qualities are adequate for video analytics as it does not impact the accuracy of the object classification~\cite{Chameleon:Sigcomm}.

\fi

\begin{figure}
    \minipage[t]{0.48\linewidth}
    \centering
    \includegraphics[width=\linewidth]{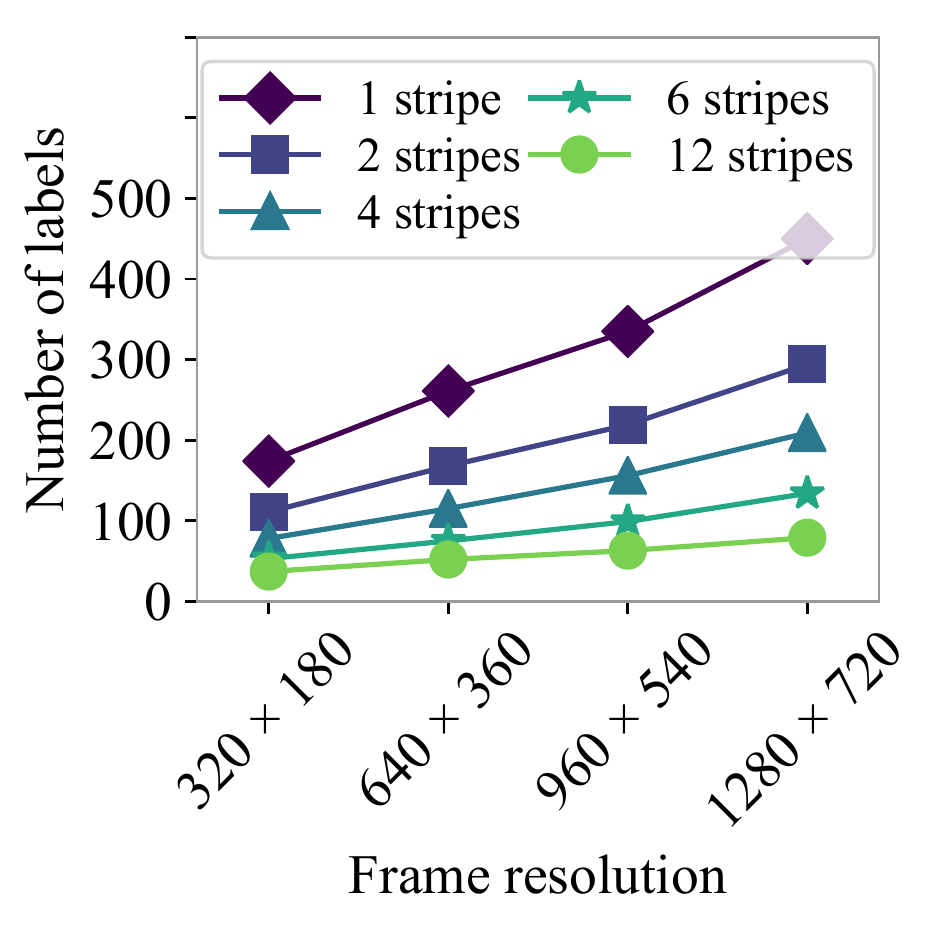}
    \caption{Number of labels for bounding box detection.
    }
    \label{fig:objdet:labels}
    \endminipage\hfill
    \minipage[t]{0.48\linewidth}
    \includegraphics[width=\linewidth]{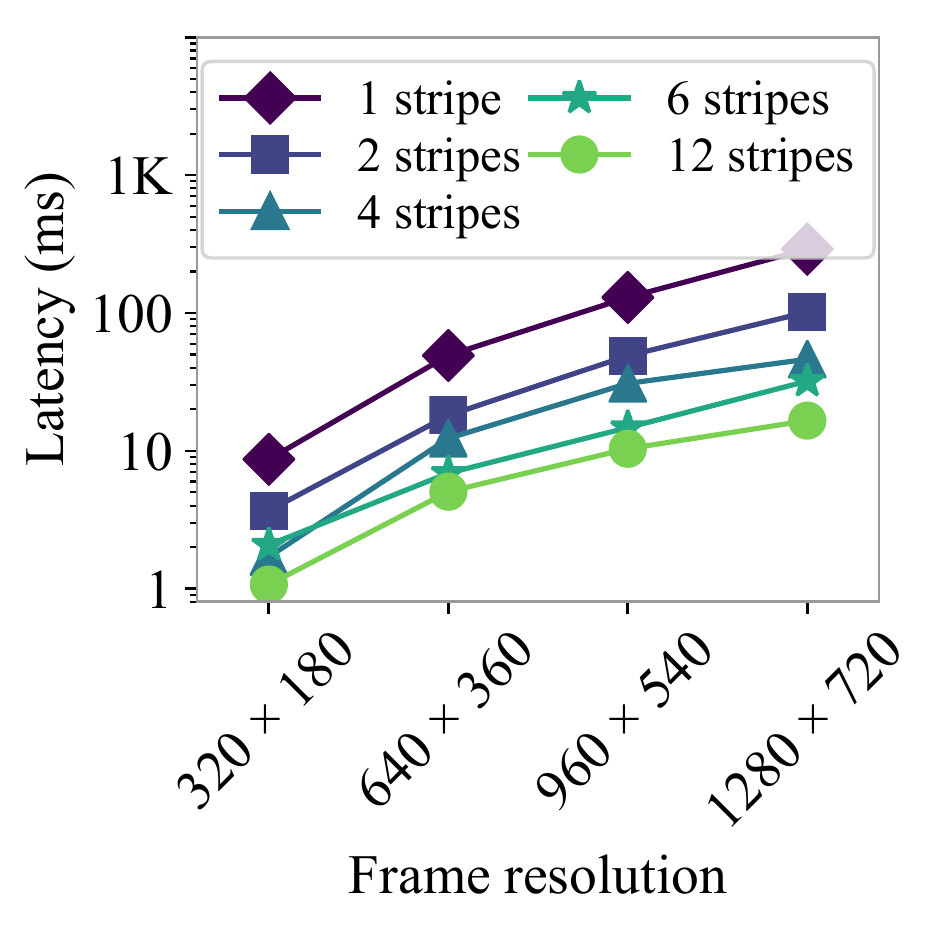}
    \caption{Latency of oblivious bounding box detection.
    }
    \label{fig:objdet:latency}
    \endminipage
\end{figure}

\begin{figure*}
    \minipage[t]{0.32\textwidth}%
    \centering
    \includegraphics[width=\linewidth]{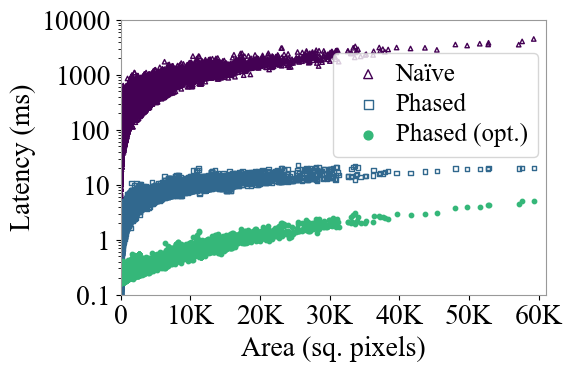}
    \caption{Oblivious object cropping.}
    \label{fig:objcrop:latency}
    \endminipage\hfill
    \minipage[t]{0.32\textwidth}
    \centering
    \includegraphics[width=\linewidth]{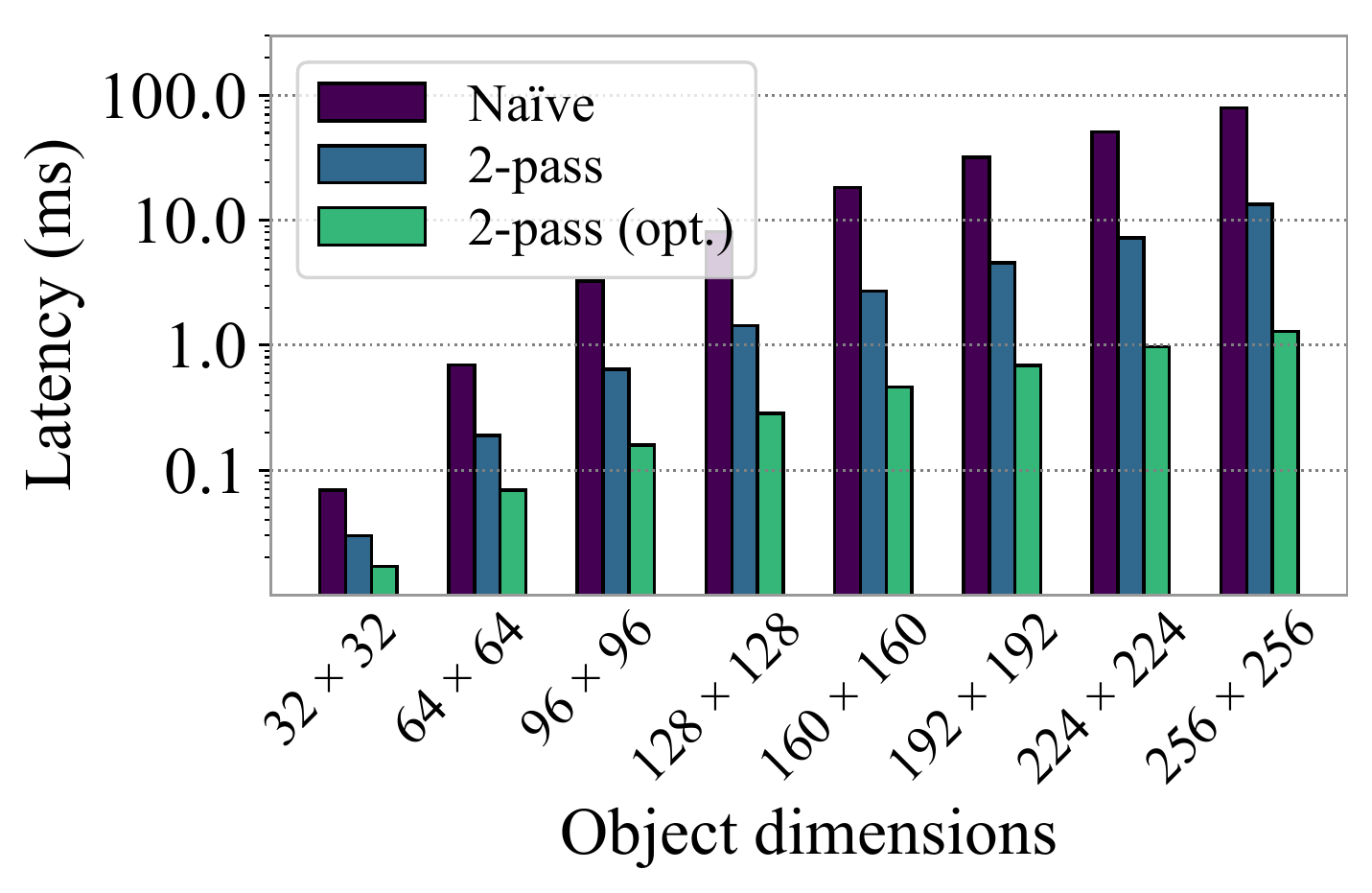}
    \caption{Oblivious object resizing.}
    \label{fig:resize:lat}
    \endminipage\hfill
    \minipage[t]{0.32\textwidth}
    \centering
    \includegraphics[width=\linewidth]{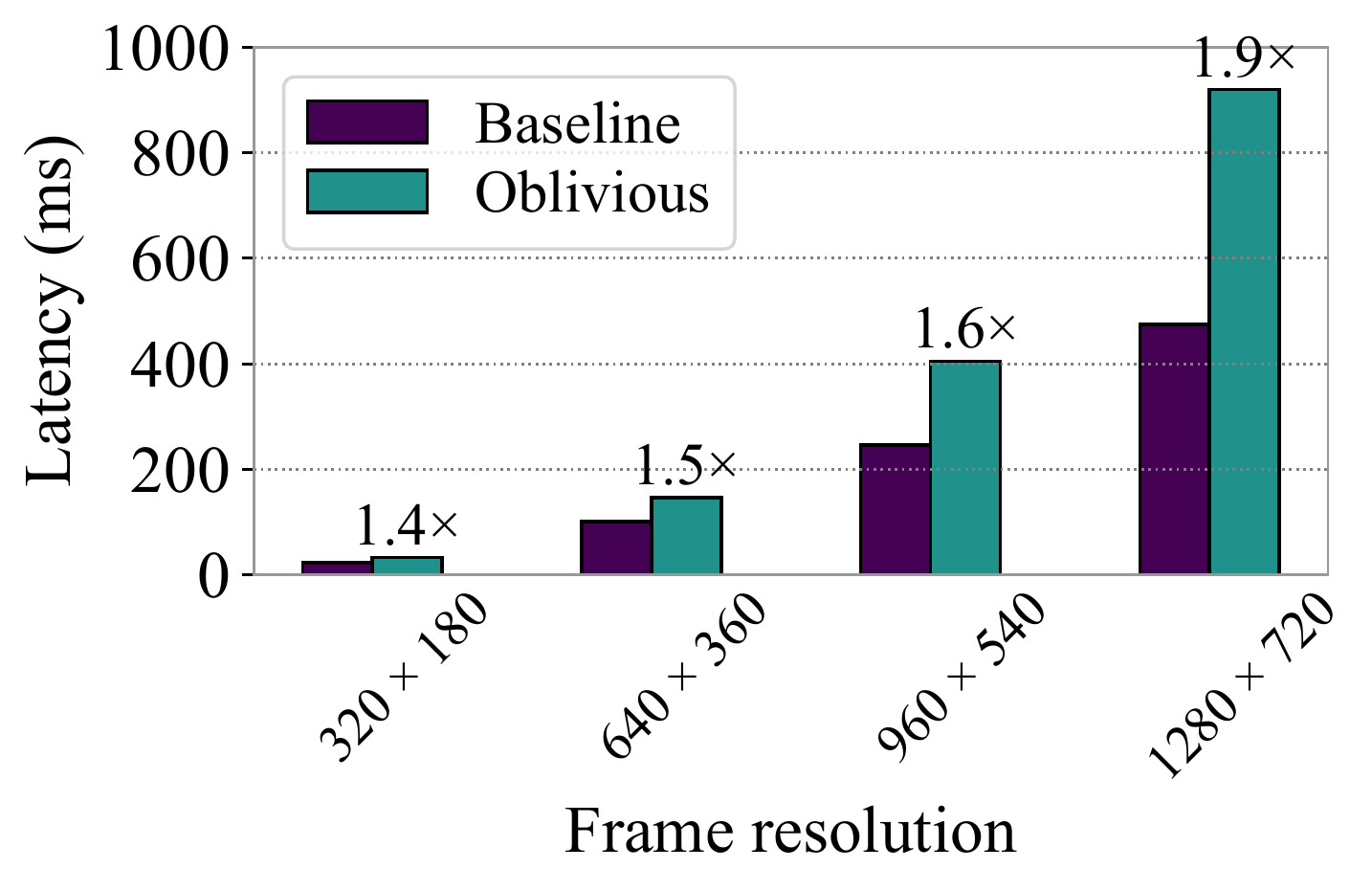}
    \caption{Oblivious object tracking.}
    \label{fig:sift:lat}
    \endminipage\hfill
\end{figure*}

\subsubsection{Background subtraction}
We set the maximum number of Gaussian components per pixel $M_\code{max} = 4$, following prior work~\cite{MOG2:Zivkovic:2004, MOG2:Zivkovic:2006}. 
Our changes for obliviousness enable us to make use of SIMD instructions for updating the Gaussian components in parallel.
This is because we now maintain the same number of components per pixel, and update operations for each component are identical.

\cref{fig:bgs:latency} plots the overhead of obliviousness on background subtraction across different resolutions. The SIMD implementation increases the latency of the routine only by $1.8\times$ over the baseline non-oblivious routine. As the routine processes each pixel in the frame independent of the rest, its latency increases linearly with the total number of pixels. 

\subsubsection{Bounding box detection}
For non-oblivious bounding box detection, we use 
the border-following algorithm of Suzuki and Abe~\cite{Suzuki:1985} (per \cref{s:algorithms:objdet}); this algorithm is efficient, running in sub-millisecond latencies.

The performance of our oblivious bounding box detection algorithm is governed by two parameters:  \begin{enumerate*}[($i$)] \item the number of stripes used in the divide-and-conquer approach, which controls the degree of parallelism, and \item an upper bound $L$ on the maximum number of labels possible per stripe, which determines the size of the algorithm's data structures.\end{enumerate*}

\cref{fig:objdet:labels} plots $L$ for streams of different frame resolutions while varying the number of stripes into which each frame is divided. As expected, as the number of stripes increases, the value of $L$ required per stripe decreases. Similarly, lower resolution frames require smaller values of $L$.

\cref{fig:objdet:latency} plots the latency of detecting all bounding boxes in a frame based on the value of the parameter $L$, ranging from a few milliseconds to hundreds of milliseconds. For a given resolution, the latency decreases as the number of stripes increase, due to two reasons:  \begin{enumerate*}[($i$)] \item increased parallelism, and \item smaller sizes of $L$ required per stripe.\end{enumerate*} Overall, the divide-and-conquer approach reduces latency by an order of magnitude down to a handful of milliseconds.

\subsubsection{Object cropping}
\label{s:eval:cropping}
We first evaluate oblivious object cropping while leaking object sizes. We include three variants: the na\"ive approach; the two-phase approach; and 
a further optimization that advances the sliding window forward multiple rows/columns at a time.
\cref{fig:objcrop:latency} plots the cost of cropping variable-sized objects from a $1280 \times 720$ frame, showing that the proposed refinements reduce latency by three orders of magnitude .

\cref{fig:resize:lat} plots the latency of obliviously resizing the target ROI within a cropped image to hide the object's size. 
While the latency of na\"ive bilinear interpolation is high (10s of milliseconds) for large objects, the optimized two-pass approach (that exploits cache locality by transposing the image before the second pass; \cref{s:algorithms:cropping:size}) reduces latency by two orders of magnitude down to one millisecond for large objects.

\subsubsection{Object tracking}
\cref{fig:sift:lat} plots the latency of object tracking with and without obliviousness.
We examine our sample streams at various resolutions to determine upper bounds on the maximum number of features in frames.
As the resolution increases, the overhead of obliviousness increases as well because our algorithm involves an oblivious sort of the intermediate set of detected features, the cost of which is superlinear in the size of the set.
Overall, the overhead is $< 2\times$.

\subsubsection{CNN classification on GPU}
\paragraph{Buffer} \cref{fig:queue:lat} benchmarks the sorting cost as a function of the object size and the buffer size. For buffer sizes smaller than 50, the sorting cost remains under \SI{5}{\milli\second}. %

\paragraph{Inference} We measure the performance of CNN object classification on the GPU. As discussed in \cref{s:algorithms:cnn}, oblivious inference comes free of cost. %
\cref{table:data:cnn} lists the throughput of different CNN models using the proprietary NVIDIA driver, with CUDA version 9.2. Each model takes as input a batch of 10 objects of size $224\times224$. %
Further, since GPU memory is limited to \SI{3}{\giga\byte}, we also list the maximum number of concurrent models that can run on our testbed.
As we show in \cref{s:eval:overall}, the latter has a direct bearing on the number of video analytics pipelines that can be concurrently served.%

\begin{figure*}[t]
\minipage[b]{0.36\textwidth}
    \centering
    \small
    \begin{tabular}[t]{r|c|c}
    \thickhline
    CNN & Batches/s & Max no. of models \T\B \\
    \hline
    AlexNet & 40.3 & ~7 \T \\
    ResNet-18 & 18.4 & 4 \\
    ResNet-50 & 8.2 &  1 \\
    VGG-16 & 5.4 & 1 \\
    VGG-19 & 4.4 & 1 \\
    Yolo & 3.9 & ~1 \B \\
    \thickhline
    \end{tabular}
    \caption{CNN throughput (batch size 10).}
    \label{table:data:cnn}
\endminipage\hfill
\minipage[b]{0.30\textwidth}%
    \centering
    \includegraphics[width=0.95\linewidth]{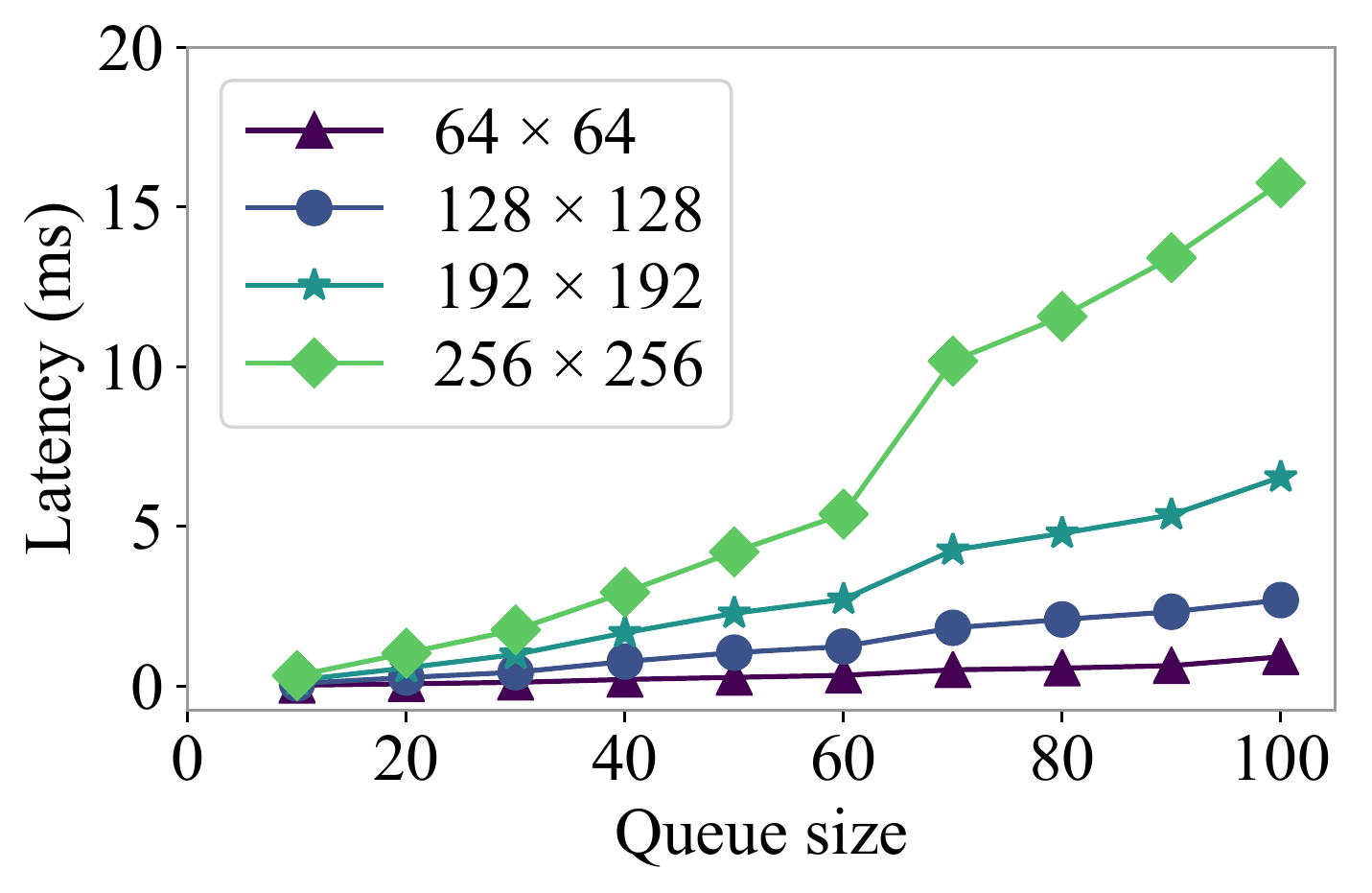}
    \caption{Oblivious queue sort.}
    \label{fig:queue:lat}
\endminipage\hfill
\minipage[b]{0.33\textwidth}
\centering
\ifpadding
    \includegraphics[width=\linewidth]{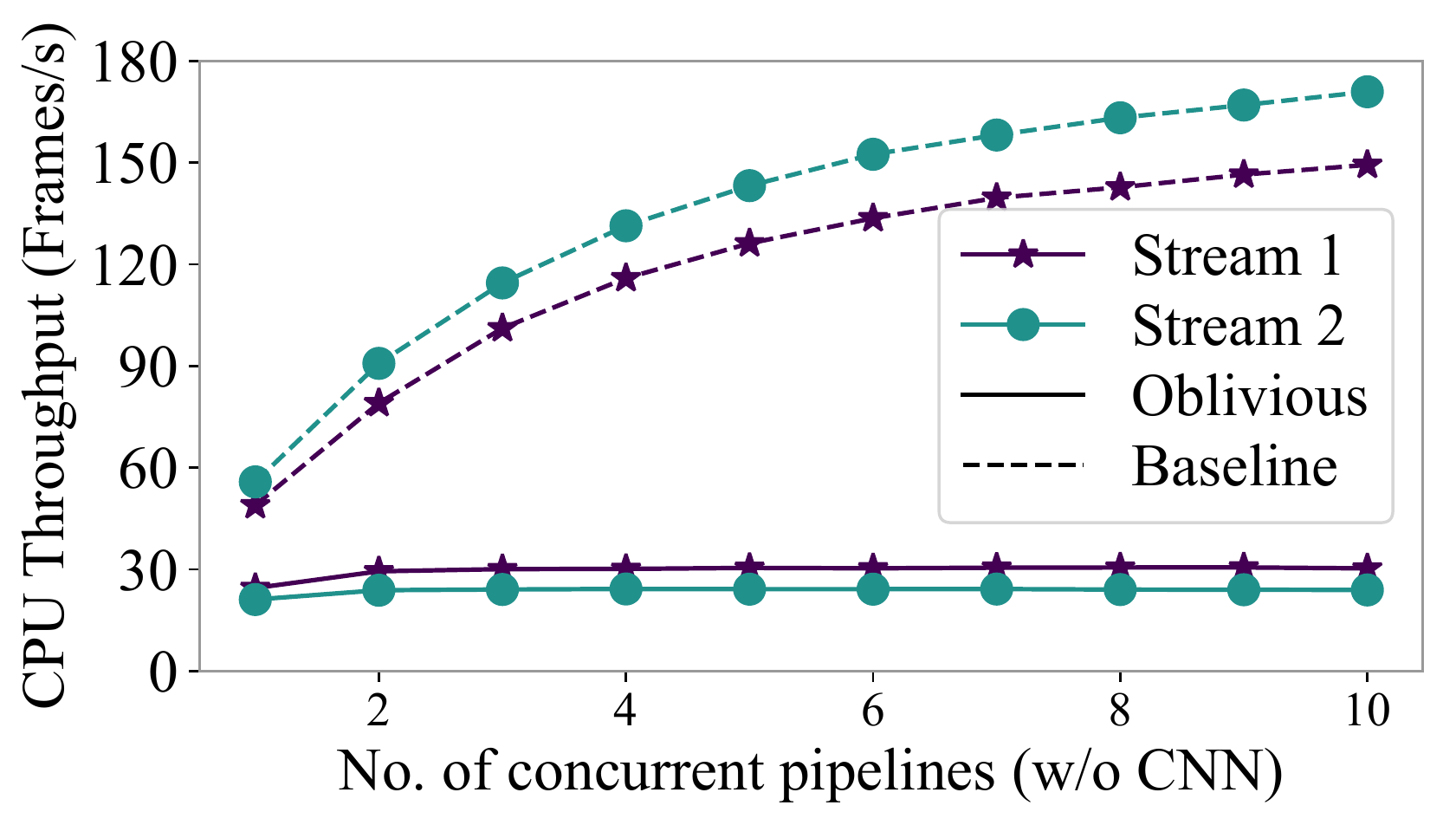}
\else
    \includegraphics[width=\linewidth]{plots/tput_cpu_no_padding_5objs.pdf}
\fi
    \caption{CPU throughput (pipeline 1).}
    \label{fig:tput:cpu}
\endminipage\hfill
\end{figure*}

\subsection{System Performance}\label{s:eval:overall}
We now evaluate the end-to-end performance of the video analytics pipeline using four real video streams. We present the overheads of running \sys's data-oblivious techniques and hosting the pipeline in a hybrid enclave.
We evaluate the two example pipelines in \cref{fig:pipeline}: pipeline 1 uses an object classifier CNN; pipeline 2 uses an object detector CNN (Yolo), and performs object tracking on the CPU.

\emph{Pipeline 1 configuration.} We run inference on objects that are larger than 1\% of the frame size as smaller detected objects do not represent any meaningful value. Across our videos, the number of such objects per frame is small---no frame has more than 5 objects, and 97-99\% of frames have less than 2 to 3 objects. 
Therefore, we configure:  \begin{enumerate*}[($i$)] \item \sys's object detection stage to conservatively output 5 objects per frame (including dummies) into the buffer, \item the consumption rate of \sys's CNN module to 2 or 3 objects per frame (depending on the stream), and \item the buffer size to 50, which suffices to prevent non-dummy objects from being overwritten.
\end{enumerate*}

\emph{Pipeline 2 configuration.} The Yolo object detection CNN ingests entire frames, instead of individual objects. In the baseline, we filter frames that don't contain any objects using background subtraction. However, we forego this filtering in the oblivious version since most frames contain foreground objects in our sample streams.
Additionally, Yolo expects the frames to be of resolution $448\times 448$. So we resize the input video streams to be of the same resolution.

\begin{figure*}
	\minipage[t]{0.33\textwidth}%
	\centering
	\ifpadding
		\includegraphics[width=\linewidth]{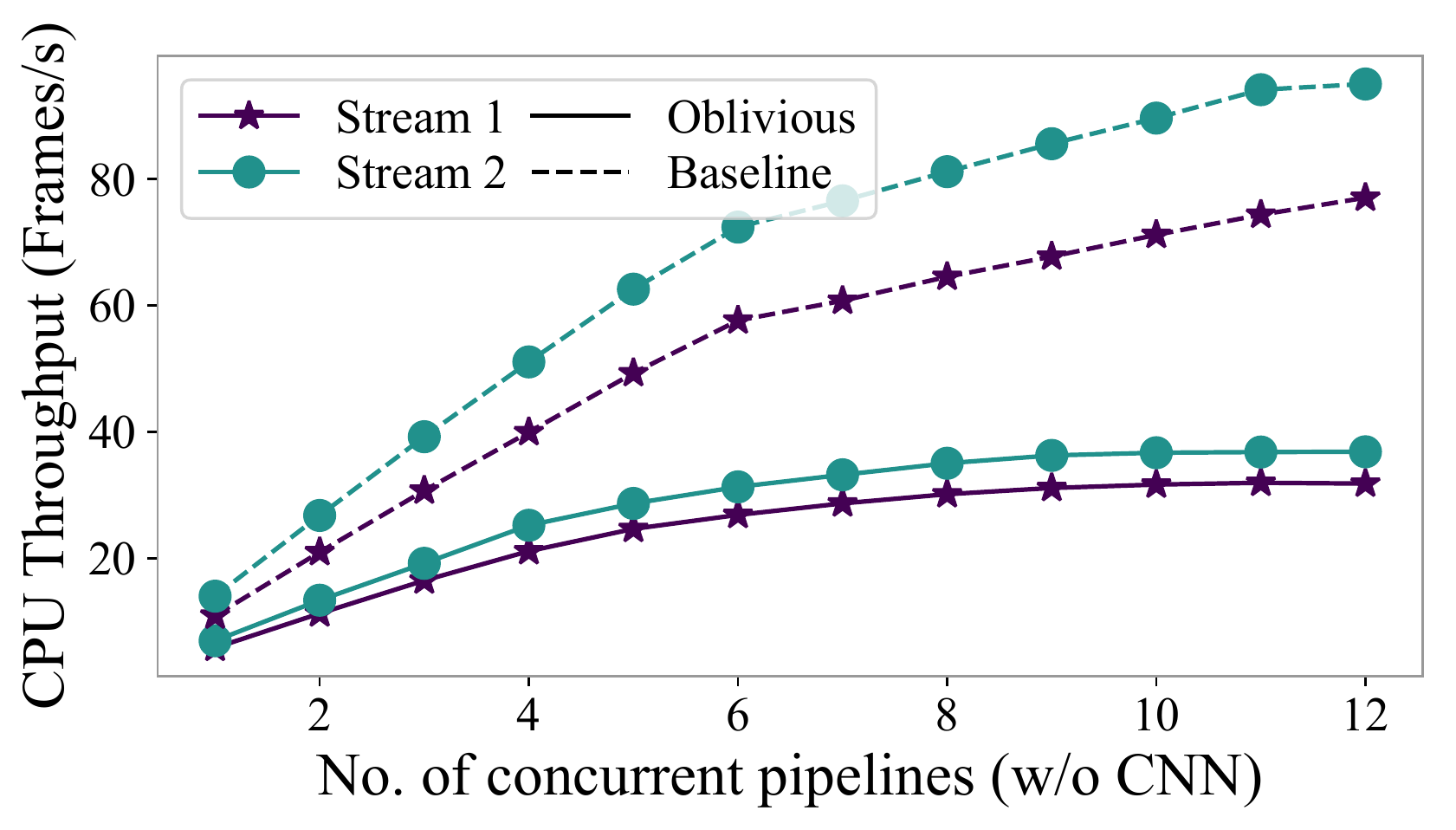}
	\else
		\includegraphics[width=\linewidth]{plots/tput_cpu_tracker.pdf}
	\fi

	\caption{CPU throughput (pipeline 2).}
	\label{fig:tput:cpu:tracker}
	\endminipage\hfill
	\minipage[t]{0.42\textwidth}
	\centering
	\ifpadding
		\includegraphics[width=\linewidth]{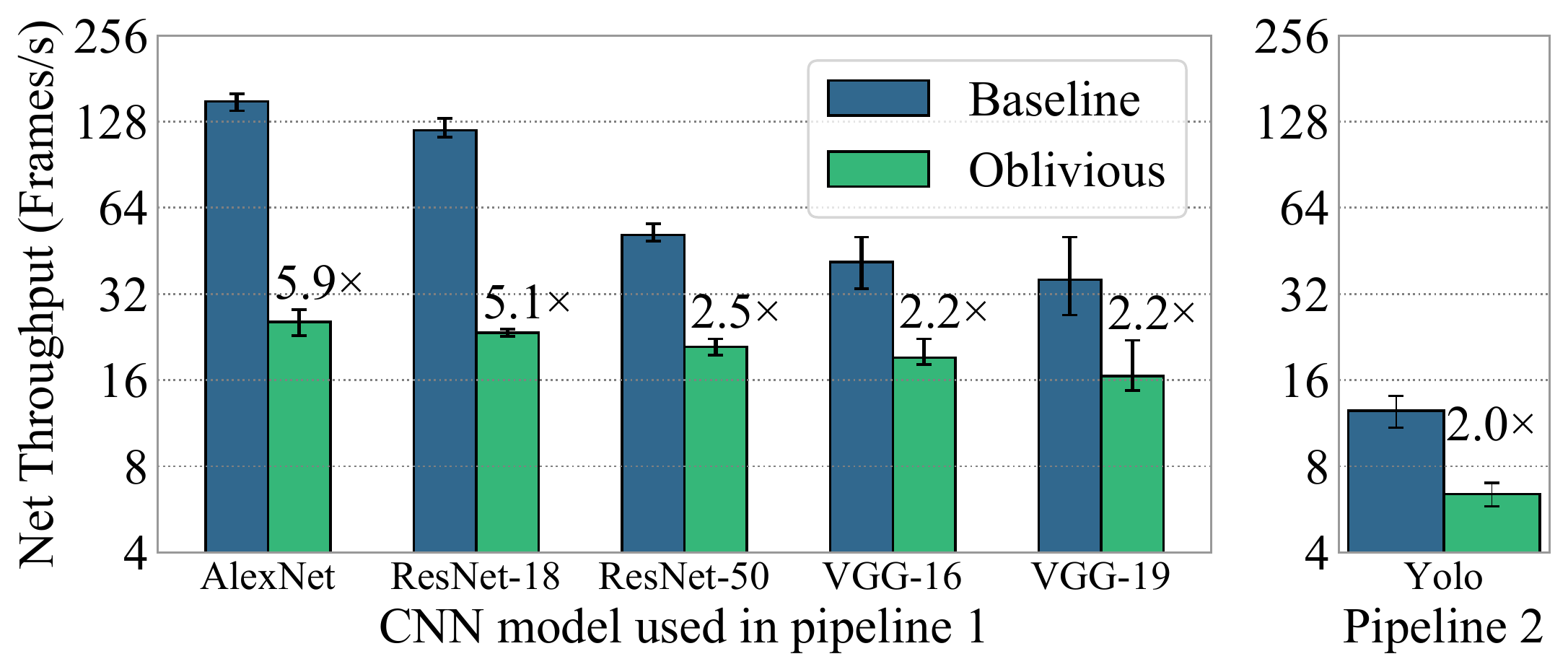}
	\else
		\includegraphics[width=\linewidth]{plots/e2e_tput_no_padding.pdf}
	\fi
	\caption{Overall pipeline throughput.}
	\label{fig:tput:e2e}
	\endminipage\hfill
	\minipage[t]{0.23\textwidth}
	\centering
	\ifpadding
		\includegraphics[width=\linewidth]{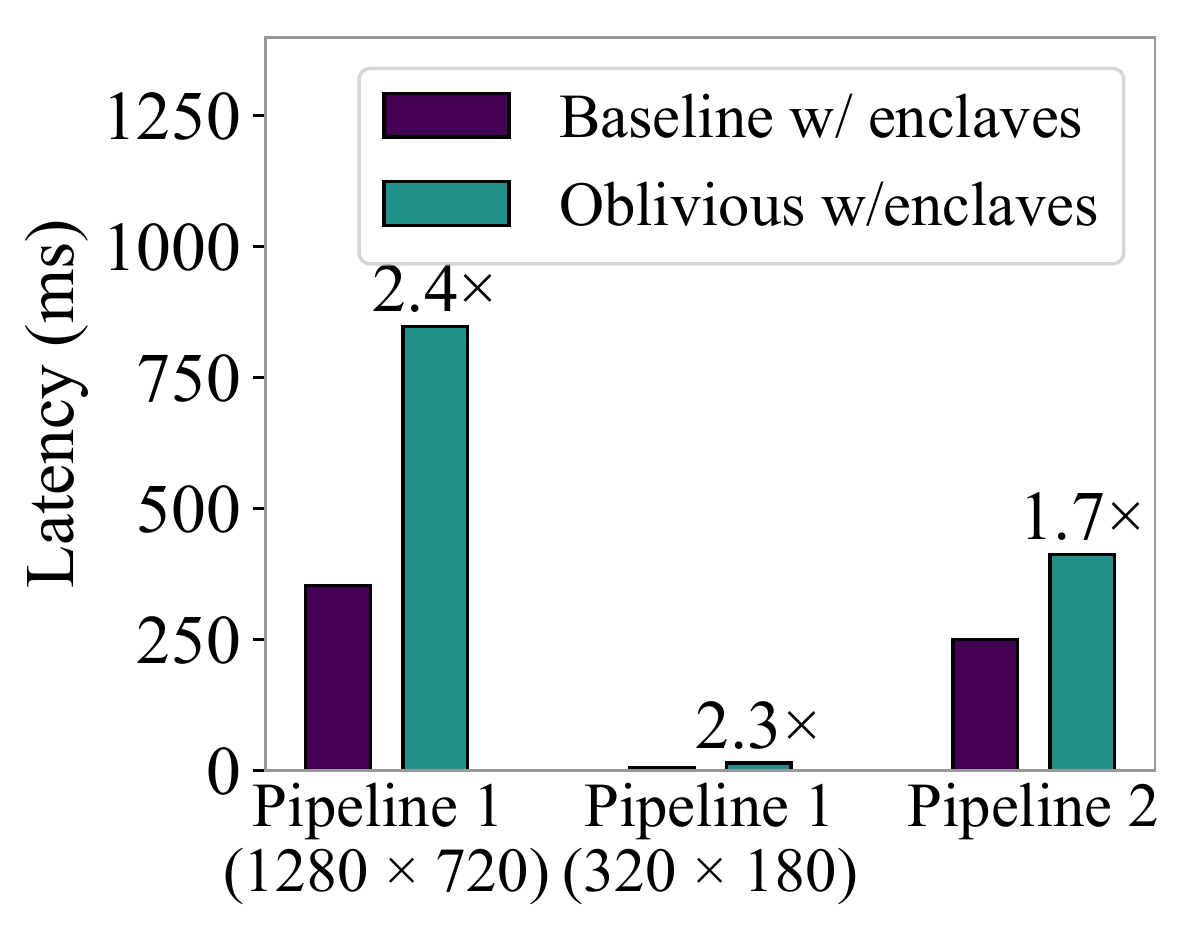}
	\else
		\includegraphics[width=\linewidth]{plots/sgx_obliv_lat_no_padding.pdf}
	\fi
	\caption{Cost of enclaves.}
	\label{fig:lat:sgx}
	\endminipage\hfill
\end{figure*}

\paragraph{Cost of obliviousness} \cref{fig:tput:cpu,fig:tput:cpu:tracker} plot the overhead of \sys on the CPU-side components of pipelines 1 and 2, while varying the number of concurrent pipelines. 
\ifpadding
    \sys reduces \emph{peak} CPU throughput by \approx$2.6\times$--$6\times$ across the two pipelines, compared to the non-oblivious baseline.  
    However, the throughput of the system ultimately depends on the number of models that can fit in GPU memory.

    \cref{fig:tput:e2e} plots \sys's end-to-end performance for both pipelines, across all four sample video streams. 
    In the presence of CNN inference, \sys's overheads depend on the model complexity. 
    Pipelines that utilize light models, such as AlexNet and ResNet-18, are bottlenecked by the CPU.
    In such cases, the overhead is determined by the cost of obliviousness incurred by the CPU components. 
    With heavier models such as ResNet-50 and VGG, the performance bottleneck shifts to the GPU. In this case, the overhead of \sys is governed by the amount of dummy objects processed by the GPU (as described in \cref{s:system:communication}).
    Overall, the cost of obliviousness remains in the range of $2.2\times$--$5.9\times$ across video streams for the first pipeline.
    In the second pipeline, the overhead is \approx$2\times$.
    The GPU can fit only a single Yolo model. 
    The overall performance, however, is bottlenecked at the CPU because the object tracking routine is relatively expensive.

    \paragraph{Cost of enclaves} We measure the cost of running the pipelines in CPU/GPU enclaves by replacing the NVIDIA stack with Graviton's stack, which comprises open-source CUDA runtime (Gdev~\cite{Kato:gdev}) and GPU driver (Nouveau~\cite{nouveau}). 

    \cref{fig:lat:sgx} compares \sys against a non-oblivious baseline when both systems are hosted in CPU/GPU enclaves. As SGX's EPC size is limited to \SI{93.5}{\mega\byte}, workloads with large memory footprints incur high overhead. 
    For pipeline 1, and for large frame resolutions, the latency of background subtraction increases from \approx\SI{6}{\milli\second} to \SI{225}{\milli\second} due to its working set size being \SI{132}{\mega\byte}.
    In Visor, the pipeline's net latency increases by $2.4\times$ (as SGX overheads mask some of \sys's overheads) while increasing the memory footprint to \SI{190}{\mega\byte}.
    When the pipeline operates on lower frame resolutions, such that its memory footprint fits within current EPC, the latency of the non-oblivious baseline tracks the latency of the insecure baseline (a few milliseconds); the additional overhead of obliviousness is $2.3\times$. %

    For pipeline 2, the limited EPC increases the latency of object tracking from \approx\SI{90}{\milli\second} to \approx\SI{240}{\milli\second}. 
    With \sys's obliviousness, the net latency increases by $1.7\times$.
\fi

\subsection{Comparison against Prior Work}\label{s:evaluation:related}
We conclude our evaluation by comparing \sys against Obfuscuro~\cite{Obfuscuro}, a state-of-the-art general-purpose system for oblivious program execution.

The current implementation of Obfuscuro supports a limited set of instructions, and hence cannot run the entire video analytics pipeline.
On this note, we ported the OpenCV object cropping module to Obfuscuro, which requires only simple assignment operations.
Cropping objects of size $128\times128$ and $16\times16$ (from a $1280\times720$ image) takes $8.5$ hours and $8$ minutes in Obfuscuro respectively, versus \SI{800}{\micro\second} and \SI{200}{\micro\second} in \sys; making \sys faster by over $6$ to $7$ orders of magnitude.
We note, however, that Obfuscuro targets stronger guarantees than \sys as it also aims to obfuscate the programs; hence, it is not a strictly apples-to-apples comparison. Nonetheless, the large gap in performance is hard to bridge, and our experiments demonstrate the benefit of \sys's customized solutions.

Other tools for automatically synthesizing or executing oblivious programs are either closed-source~\cite{Raccoon,Wu:ISSTA:2018}, require special hardware~\cite{HOP:Nayak, ghostrider, Phantom}, or require custom language support~\cite{Fact:Cauligi}.
However, we note that the authors of Raccoon~\cite{Raccoon} (which provides similar levels of security as \sys) report up to $1000\times$ overhead on toy programs; the overhead would arguably be higher for complex programs like video analytics.

\section{Discussion}\label{s:discussion}

\paragraph{Attacks on upper bounds}
For efficiency, \sys extracts a fixed number of objects per frame based on a user-specified upper bound.
However, this leaves \sys open to adversarial inputs: an attacker who knows this upper bound can attempt to confuse the analytics pipeline by operating many objects in the frame at the same time.

To mitigate such attacks, we suggest two potential strategies:
\begin{enumerate*}[(i)]
    \item For frames containing $>=N$ objects (as detected in \cref{s:algorithms:objdet}), process those frames off the critical path using worst-case bounds (\eg total number of pixels). While this approach leaks which specific frames contain $>=N$ objects, the leakage may be acceptable considering these frames are suspicious.

    \item Filter objects based on their properties like object size or object location: \eg for a traffic feed, only select objects at the center of the traffic intersection. This limits the number of valid objects possible per frame, raising the bar for mounting such attacks. One can also apply richer filters on the pipeline results and reprocess frames with suspicious content.
\end{enumerate*}

\paragraph{Oblivious-by-design encoding}
Instead of designing oblivious versions of existing codecs, it may be possible to construct an oblivious-by-design coding scheme that is \begin{enumerate*}[(i)]\item potentially simpler, and \item performs better than \sys's oblivious decoding. \end{enumerate*}
This alternate design point is an interesting direction for future work.
We note, however, that any such codec would need to produce a perfectly constant bitrate (CBR) per frame to prevent bitrate leakage over the network.
While CBR codecs have been explored in the video literature, they are inferior to variable bitrate schemes (VBR) such as VP8 because they are lossier.
In other words, an oblivious CBR scheme would consume greater bandwidth than VP8 to match its video quality (and therefore, VP8 with padding), though it may indeed be simpler.
In \sys, we optimize for quality.

\section{Related Work}\label{s:relatedwork}

To the best of our knowledge, \sys is the first system for the secure execution of vision pipelines. We discuss prior work related to various aspects of \sys.

\paragraph{Video processing systems} A wide range of optimizations have been proposed to improve the efficiency of video analytic pipelines~\cite{Focus:OSDI18, NoScope, Chameleon:Sigcomm, VideoStorm}. %
These systems offer different design points for enabling trade-offs between performance and accuracy.
Their techniques are complementary to \sys which can benefit from their performance efficiency.

\paragraph{Data-oblivious techniques}
Eppstein \etal~\cite{Eppstein:2010:Oblivious:Geometric} develop data-oblivious algorithms for geometric computations.
Ohrimenko \etal~\cite{Ohrimenko:ObliviousML} propose data-oblivious machine learning algorithms running inside CPU TEEs.
These works are similar in spirit to \sys, but are not applicable to our setting.

Oblivious RAM~\cite{ORAM:GR96} is a general-purpose cryptographic solution for eliminating access-pattern leakage.
While recent advancements have reduced its computational overhead~\cite{StefanovDSFRYD13}, it still remains several orders of magnitude more expensive than customized solutions.
Oblix~\cite{Oblix} and Zerotrace~\cite{Zerotrace:Sasy} enable ORAM support for applications running within hardware enclaves, but have similar limitations.

Various systems~\cite{Raccoon, ghostrider, Obfuscuro, PAO:Sinha:2017, HOP:Nayak, Phantom, Wu:ISSTA:2018, Fact:Cauligi} also offer generic solutions for hiding access patterns at different levels, with the help of ORAM, specialized hardware, or compiler-based techniques.
Generic solutions, however, are less efficient than customized solutions (such as \sys) which can exploit algorithmic patterns for greater efficiency.

\paragraph{Side-channel defenses for TEEs}
\sys provides systemic protection against attacks that exploit access pattern leakage in enclaves.
Systems for data-oblivious execution (such as Obfuscuro~\cite{Obfuscuro} and Raccoon~\cite{Raccoon}) provide similar levels of security for general-purpose workloads, while \sys is tailored to vision pipelines.

In contrast, a variety of defenses have also been proposed to detect~\cite{dejavu} or mitigate \emph{specific} classes of access-pattern leakage.
For example, Cloak~\cite{cloak}, Varys~\cite{SGX:defense:Varys}, and Hyperrace~\cite{SGX:defense:Hyperrace} target cache-based attacks; while T-SGX~\cite{TSGX} and Shinde \etal~\cite{Shinde2016} propose defenses for paging-based attacks.
DR.SGX~\cite{Brasser:DRSGX} mitigates access pattern leakage by frequently re-randomizing data locations, but can leak information if the enclave program makes predictable memory accesses.

Telekine~\cite{telekine} mitigates side-channels in GPU TEEs induced by CPU-GPU communication patterns, similar to \sys's oblivious CPU-GPU communication protocol (though the latter is specific to \sys's use case).

\drop{
\paragraph{SGX side-channel defenses}
A variety of defenses have been proposed to detect~\mbox{\cite{dejavu}} or mitigate side-channel attacks on SGX, such as cache attacks~\mbox{\cite{cloak, SGX:defense:Varys}} and paging-based attacks~\mbox{\cite{TSGX, Shinde2016}}. These solutions, however, do not provide systemic protection against leakage due to access patterns, and only focus on mitigating specific side-channels.
DR.SGX~\mbox{\cite{Brasser:DRSGX}} aims to mitigate access pattern leakage by frequently re-randomizing data locations, but still leaks information if the enclave program makes predictable memory accesses.
}

\paragraph{Secure inference}
Several recent works propose cryptographic solutions for CNN inference~\cite{Gazelle, Minionn, CryptoNets, Deepsecure, Riazi:Chameleon:ML} relying on homomorphic encryption and/or secure multi-party computation \cite{Yao:GC}. While cryptographic approaches avoid the pitfalls of TEE-based CNN inference, the latter remains faster by orders of magnitude~\cite{Slalom:Tramer, Chiron:Hunt}.

\section{Conclusion}
We presented \sys, a system that enables privacy-preserving video analytics services. \sys uses a hybrid TEE architecture that spans both the CPU and the GPU, as well as novel data-oblivious vision algorithms. \sys provides strong confidentiality and integrity guarantees, for video streams and models, in the presence of privileged attackers and malicious co-tenants. Our implementation of \sys shows limited performance overhead for the provided level of security.

\section*{Acknowledgments}
We are grateful to Chia-Che Tsai for helping us instrument the Graphene LibOS.
We thank our shepherd, Kaveh Razavi, and the anonymous reviewers for their insightful comments. 
We also thank Stefan Saroiu, Yuanchao Shu, and members of the RISELab at UC Berkeley for helpful feedback on the paper. 
This work was supported in part by the NSF CISE Expeditions Award CCF-1730628, and gifts from the Sloan Foundation, Bakar Program,  Alibaba, Amazon Web Services, Ant Financial, Capital One, Ericsson, Facebook, Futurewei, Google, Intel, Microsoft, Nvidia, Scotiabank, Splunk, and VMware.

{\footnotesize
\begin{flushleft}
\setlength{\parskip}{0pt}
\setlength{\itemsep}{0pt}
\bibliographystyle{abbrv}
\bibliography{bib/references,bib/str_short,bib/conf}
\end{flushleft}
}
\appendix
\ifextended
\section{Security Proofs and Pseudocode}\label{s:proofs}

We now provide detailed pseudocode along with proofs of security for our algorithms.
We start by providing a formal definition of data-obliviousness.

\newcommand{\algo}{\mathcal{A}}
\newcommand{\simulator}{\mathcal{S}}
\newcommand{\leakage}{\mathcal{L}}
Let $\code{trace}_\algo(x)$ be the trace of observations that the adversary can make during an execution of an algorithm $\algo$, when given some input $x$, \ie the sequence of accessed memory addresses, along with the (encrypted) data that is read or written to the addresses.
To prove that the algorithm is data-oblivious, we show that there exists a simulator program that can produce a trace $T$ indistinguishable from $\code{trace}_\algo (x)$, when given \emph{only} the public parameters for the algorithm, and regardless of the value of $x$.
Since $T$ does not depend on any private data, the indistinguishability of $T$ and $\code{trace}_\algo$ implies that the latter leaks no information about the private data to the adversary, and only depends on the public parameters.
The following definition captures the definition formally.

\begin{definition}[Data-obliviousness]\label{def:oblivious}
We say than an algorithm $\algo$ is data-oblivious if there exists a polynomial-time simulator $\simulator$ such that for any input $x$
$$
\code{trace}_\algo(x) \equiv \simulator(\leakage(\algo))
$$
where $\leakage(\algo)$ is the leakage function and represents the public parameters of $\algo$.
\end{definition}

We now prove the security of each of our algorithms with respect to \cref{def:oblivious} in the following subsections.
\cref{table:leakage} summarizes the public parameters across \sys oblivious vision modules that are leaked to the attacker.

\begin{figure}
\centering
    \small
    \begin{tabular}[t]{p{2.5cm}|p{5cm}}
    \thickhline
    Component & Public parameters \T\B \\
    \thickhline
    Video decoding &
\begin{enumerate*}[($i$)]
    \item Metadata of video stream (format, frame rate, resolution);
    \item Number of bits used to encode each (padded) row of blocks;
    \item Maximum number of bits encoded per 2-byte chunk. 
\end{enumerate*}\\\hline
    Background sub. & -- \\\hline
    Bounding box det. &  
    \begin{enumerate*}[($i$)]
        \item Maximum number of objects per image; \item Maximum number of different labels that can be assigned to pixels (an object consists of all labels that are adjacent to each other).
    \end{enumerate*}\\\hline
    Object cropping & 
    Upper bounds on object dimensions.
    \\\hline
    Object tracking & 
    \begin{enumerate*}[($i$)]
    \item An upper bound on the intermediate number of features;
    \item An upper bound on the total number of features.
    \end{enumerate*}\\\hline
    CNN Inference & CNN architecture. \\\hline
    Overall & Modules and algorithms used in the pipeline.\\
    \thickhline
    \end{tabular}
    \caption{Summary of public parameters in \sys's oblivious vision modules observable by the attacker. 
    These consist of the input parameters provided to \sys, along with information leaked by \sys (such as frame rate and resolution).
    }
    \label{table:leakage}
    \vspace{0.5cm}
\end{figure}

\subsection{Oblivious video decoding}\label{s:proofs:decoder}

\cref{alg:decoding} provides detailed pseudocode for oblivious decoding the bitstream into pixel coefficients during video decoding, as described in \cref{s:decoding:bitstream}. We first explain the pseudocode in more detail, following the high-level description of \cref{s:decoding:bitstream}.

In our implementation, we model the prefix tree as a finite state machine (FSM).
While traversing the tree, we decode a single bit at each state (\ie each node in the tree) using the \textsc{EntropyDecode} subroutine.
This subroutine takes in a pointer \code{ptr} to the bitstream, and decodes a single bit from the bitstream via arithmetic operations. If no more bits can be decoded at the current position, it outputs \mynull; otherwise, it outputs the decoded bit 0 or 1.

To enable decoding and traversal, each state $S$ (\ie each node in the tree) is a structure containing four fields:
$(\code{prob}, \allowbreak \code{next}_0, \allowbreak \code{next}_1, \code{type})$. 
Here, $\code{prob}$ is the probability that the bit to be decoded at $S$ is 0 (as defined in the VP8 specifications~\cite{RFC:VP8Decoding}); and $\code{next}_0$ and $\code{next}_1$ are the indices of the next state based on whether the decoded bit is a 0 or 1.
Some states in the FSM are end states, \ie states that complete reconstructing a data value; for these states, \code{type} is set to \code{`end'}.
States that are not end states (\ie decode intermediate bits) have \code{type} set to \code{`mid'}.
The FSM also contains a dummy state $S_\code{dummy}$ that performs dummy bit decodes by invoking the entropy decoder with \dummy set to true; for the dummy state, \code{type} is set to \code{`dummy'}.

Next, we represent the FSM as an array---\code{Nodes} in \cref{alg:decoding}.
We set $\code{Nodes}[0]$ to be $S_\code{dummy}$, and $\code{Nodes}[1]$ to be the starting state. 
This enables us to implement transitions to any state $S_j$ by simply fetching the state at index $j$ from the array using the \oarr primitive.
As a result, the current state of the FSM remains hidden across transitions.
Each transition passes four items of information to the next state: (i)~the state that made the transition $S_\code{parent}$; (ii)~an integer \code{pos} that denotes the position in the bitstream of the current bit being decoded; (iii)~the (partially) constructed data value \code{data}, and (iv)~a counter \code{ctr} that counts the number of bits decoded at each position.
Note that the structure of the prefix tree (and hence the array \code{Nodes}) is public information since it is constructed per the VP8 specifications~\cite{RFC:VP8Decoding}.

\begin{algorithm}[t]
    \small
    \caption{Bitstream decoding}\label{alg:decoding}
    \begin{algorithmic}[1]
        \State \textbf{Constants:} Upper bound on number of encoded bits per 2-byte chunk in bitstream $N_\code{chunk}$; total number of bits in bitstream $N_\code{bits}$; array representation of the prefix tree for decoding \code{Nodes}
        \State \textbf{Globals:} The data value being decoded \code{data}; counter \code{ctr} that counts the number of bits decoded per chunk in the bitstream
        \State \textbf{Input:} Bitstream $B$
        \Procedure{DecodeBitstream}{$B$}
        \State $\code{ptr} = B.\code{start}$
        \State $S = \code{Nodes}[1]$
        \State $S_\code{parent} = \mynull$, $\code{data} = 0$, $\code{ctr} = 0$, $\code{pos} = 0$
        \State $O = [~]$
        \While {$\code{ptr} < B.\code{start} + N_\code{bits}$}
        \State $\code{isDummy = (S.\code{type} == `dummy')}$
        \State $b = \textsc{EntropyDecode}(\code{isDummy}, \code{ptr}, S.\code{prob})$
        \State $\code{isDummy = (b == \mynull)}$
        \State $\code{data} = \textsc{UpdateData}(\code{isDummy}, \code{data}, b)$

        \State $\code{pos} += 1$
        \State $\code{ctr} = \oselect(\code{ctr} == N_\code{chunk}, 0, \code{ctr}+1)$

        \State $\code{isEnd = (S.\code{type} == `\code{end}')}$
        \State $o_1 = \oselect(\code{isEnd}, \code{pos}, 0)$
        \State $o_2 = \oselect(\code{isEnd}, \code{data}, \mynull)$

        \Statex
        \State $\code{parent} = \oselect(\lnot\code{isDummy}, S.\code{index}, S_\code{parent}.\code{index})$
        \State $\code{next} = \oselect(b == 0, S.\code{next}_0, S.\code{index})$
        \State $\code{next} = \oselect(b == 1, S.\code{next}_1, \code{next})$
        \State $\code{next} = \oselect(\code{isDummy}, 0, \code{next})$
        \State $\code{next} = \oselect(\code{isDummy} \wedge$ 
        \Statex \hskip\algorithmicindent \hskip\algorithmicindent \hskip\algorithmicindent \hskip\algorithmicindent
        $\code{ctr} == N_\code{chunk}, \code{parent}, \code{next})$
        \Statex
        \State $O.\textsc{Append}((o_1, o_2))$

        \State $S_\code{parent} = \oarr(\code{Nodes},\code{parent})$
        \State $S = \oarr(\code{Nodes}, \code{next})$
        \State $\code{ptr} = \oselect(\code{ctr} == N_\code{chunk}, \code{ptr}+2, \code{ptr})$
        \EndWhile
        \State $\osort(O)$
        \State \Return $O$
        \EndProcedure

    \end{algorithmic}
\end{algorithm}

\begin{theorem}
    The bitstream decoding algorithm in \cref{alg:decoding} is data-oblivious per \cref{def:oblivious}, with public parameters $N_\code{bits}$, $N_\code{chunk}$, and the size of the prefix tree array \code{Nodes} (which is a known constant).
\end{theorem}

\begin{proof}
    The simulator starts be generating a random bitstream $B$ of length $N_\code{bits}$, and then simply executes \cref{alg:decoding}.
    It outputs the trace produced by the algorithm.

    Lines 5--8 have fixed access patterns.

    The loop in line 9 runs a fixed number of times: \code{ctr} increments by 1 in each run of the loop on line 15 until it becomes equal to $N_\code{chunk}$, at which point the loop variable \code{ptr} is incremented by 2 in line 27.
    Thus, the loop makes exactly $N_\code{bits}\times N_\code{chunk} / 2$ iterations.

    Within the loop, line 10 has fixed access patterns.
    In line 11, the function \textsc{EntropyDecode} dereferences \code{ptr} and decodes a single bit from the dereferenced value using simple arithmetic operations; if \code{isDummy} is true it performs dummy operations instead, using \oselect. 
    Its access patterns thus only depend on the location pointed to by \code{ptr} within the bitstream $B$, and not the contents of the bitstream.
    Further, the value of the loop variable \code{ptr} is incremented per a fixed schedule (as described above), and thus only depends on the value of $N_\code{bits}$.

    Line 12 has fixed access patterns.
    In line 13, the function \textsc{UpdateData} updates the value of \code{data} with $b$ using arithmetic operations implemented using \oselect. Its access patterns are thus independent of the \code{data} or $b$.

    Lines 14--24 have fixed access patterns.
    The access patterns of lines 25--26 only depend on the length of \code{Nodes}, which is fixed and public.
    Line 27 also has fixed access patterns.

    Finally, line 29 uses the \osort primitive to sort the array $O$; its access patterns thus depend on the length of $O$. Since a single tuple is appended to $O$ per iteration of the loop, the length of $O$ is equal to the number of iterations, which only depends on $N_\code{bits}$ and $N_\code{chunk}$ as described above.

    Thus, the trace produced by the algorithm can be simulated only using the values of $N_\code{bits}$ and $N_\code{chunk}$.
\end{proof}

\subsection{Oblivious image processing}\label{s:proofs:algorithms}

In this section, we provide pseudocode along with proofs of security for the image processing algorithms described in \cref{s:algorithms}. 
For each algorithm, we first briefly describe its pseudocode, and then prove its security with respect to \cref{def:oblivious}.

\subsubsection{Background subtraction}
As described in \cref{s:algorithms:bgs}, the background subtraction algorithm (\cref{alg:bgs}) maintains a mixture of $M$ Gaussian components per pixel.

Let $\vec{x}^{(t)}$ denote the value of a pixel in RGB at time $t$.
The algorithm uses the value of the pixel to update each Gaussian component via a set of arithmetic operations (lines 5--8 in the pseudocode) along with their weights $\pi_m$ such that,
over time, components that represent background values for the pixel come to have larger weights, while foreground values are represented by components having smaller weights.

Then, it compute the distance of the sample from each of the $M$ components. If no component is sufficiently close, it adds a new component, %
increments $M$, and if the new $M$ $>$ $M_\code{max}$, discards the component with the smallest weight $\pi_m$ (lines 9--21).

Finally, it uses the $B$ largest components by weight to determine whether the pixel is part of the background (lines 22--30).
Note that $M_\code{max}$ and $B$ are algorithmic constants, independent of the input video streams.

\begin{algorithm}[t]
    \small
    \caption{Background subtraction}\label{alg:bgs}
    \begin{algorithmic}[1]
        \State \textbf{Constants:} Maximum number of Gaussian components $M_\code{max}$, number of components to count towards background decision $B$; threshold measures $\delta_\code{thr}$, $c_f$, and $c_\code{thr}$
        \State \textbf{Globals:} Actual number of Gaussian components $M$, array of Gaussian components \code{GMM} of size $M_\code{max}$
        \State \textbf{Input:} Pixel $x$
        \Procedure{BackgroundSubtraction}{$x$}
        \For{$m=1$ \textbf{to} $M_\code{max}$}
        \State $\dummy = (m > M)$
        \State $\textsc{UpdateGaussian}(\dummy, \code{GMM}[m], x)$
        \EndFor

        \Statex
        \State $\textsc{SortByWeight}(\code{GMM})$

        \State $\code{isClose} = \code{false}$
        \For{$m=1$ \textbf{to} $M_\code{max}$}
        \State $\delta = \textsc{GetDistance}(\code{GMM}[m], x)$
        \State $\code{isClose} = \code{isClose} \vee (\delta > \delta_\code{thr})$
        \EndFor
        \State $M = \oselect(\code{isClose} \wedge (M < M_\code{max}), M+1, M)$
        \State $G = \textsc{GenerateGaussian}()$
        \State $\code{GMM}[M_\code{max}-1] = \oselect(\code{isClose}, \code{GMM}[M_\code{max}-1], G)$

        \For{$m=M_\code{max-1}$ \textbf{to} $1$}
        \State $\code{toSwap} = (\code{GMM}[m].\pi < \code{GMM}[m+1].\pi)$
        \State $\code{GMM}[m] = \oselect(\code{toSwap}, \code{GMM}[m+1], \code{GMM}[m])$
        \EndFor

        \Statex
        \State $c = 0$
        \State $p = 0$
        \State $\code{toInclude} = \code{true}$
        \For{$m=1$ \textbf{to} $B$}
        \State $c = c + \code{GMM}[m].\pi$
        \State $p = \oselect(\code{toInclude}, p+c, p)$
        \State $\code{toInclude} = \oselect(c > c_f, \code{false}, \code{toInclude})$
        \EndFor
        \State \Return $p > c_\code{thr}$
        \EndProcedure
    \end{algorithmic}
\end{algorithm}

\begin{theorem}
    The background subtraction algorithm in \cref{alg:bgs} is data-oblivious per \cref{def:oblivious}, with public parameters $M_\code{max}$ and $B$ (which are known constants).
\end{theorem}
\begin{proof}
    The simulator chooses a random pixel value $x$ and simply runs the algorithm.
    It outputs the trace produced by the algorithm.

    Lines 6--7 are executed exactly $M_\code{max}$ times. Here, the loop variable $m$ is public information.
    Line 6 has fixed access patterns. The function \textsc{UpdateGaussian} performs a set of arithmetic operations independent of the value of $x$, via \oselect operations using the condition value \code{isDummy}.
    The function updates 
    The access patterns of line 7 therefore only depend on $m$.

    \textsc{SortByWeight} in line 9 sorts the \code{GMM} array using the oblivious sorting primitive \osort. Hence, the access patterns of this step only depend on the length of \code{GMM}, which is $M_\code{max}$.

    Lines 12--13 are executed exactly $M_\code{max}$ times. %
    The function \textsc{GetDistance} computes the distance of $x$ from $\code{GMM}[m]$ via arithmetic operations, independent of the value of $x$.
    Thus, the access patterns of line 12 thus depends only on the loop variable $m$. 
    Line 13 has fixed access patterns.
    Lines 15--16 have fixed access patterns, and the access patterns of line 17 depends only on $M_\code{max}$.
    Lines 19--20 are executed exactly $M_\code{max}-1$ times. %
    The access patterns of both lines depend only on the loop variable $m$.

    Lines 22--24 have fixed access patterns.
    Lines 26--28 are executed exactly $M_\code{max}$ times. %
    The access patterns of line 26 depend only on the loop variable $m$; lines 27 and 28 have fixed access patterns.

    Thus, the trace produced by the simulator is indistinguishable from the trace produced by a real run.
\end{proof}

\subsubsection{Object detection}

\cref{alg:bbox} describes our algorithm for detecting bounding boxes of objects in an input image.
The algorithm maintains a list $L$ of tuples of the form $(\code{parent}, \code{bbox})$, where each tuple corresponds to a distinct ``label'' that will eventually be mapped to each blob. Initially, the list $L$ is empty. The \code{parent} field identifies other labels that are connected to the tuple's label (explained shortly), and the \code{bbox} field maintains the coordinates of the bounding box of the label (or blob).%

The algorithm first scans the image row-wise (lines 6--22).
Whenever a white pixel is detected, the algorithm checks if any of its neighbors scanned thus far were also white (\ie pixel to the left and the three pixels above).
In case at least one neighbor is white, the pixel is assigned the label of the neighbor with the smallest numerical value, $l_\code{min}$. 
The algorithm records that all white neighbors are connected, by setting the \code{parent} fields for each neighboring label to $l_\code{min}$ and updating the \code{bbox} field for $l_\code{min}$.
In case no neighbor is white, the pixel is assigned a new label, and a new entry is added to the list $L$, with its $\code{parent}$ set to the label itself and \code{bbox} as the coordinates of the current pixel.

Next, the algorithm \emph{merges} the bounding boxes of all connected labels into a single bounding box (lines 23--35).
Specifically, for every label $l$ in $L$, the algorithm first obtains the \code{parent} label of $l$ (say $l_\code{par}$), and then updates the \code{bbox} of $l_\code{par}$ to include the \code{bbox} of $l$. It repeats the process recursively with $l_\code{par}$, until it reaches a root label $l_\code{root}$ whose \code{parent} value is the label itself.
The process repeats for all labels in $L$, until only the root labels are left behind. 
Each root label corresponds to a distinct object in the frame.

\begin{algorithm}[th!]
    \small
    \caption{Bounding box detection}\label{alg:bbox}
    \begin{algorithmic}[1]
        \State \textbf{Constants:} Maximum number of labels $N$
        \State \textbf{Input:} Frame $F$
        \Procedure{BoundingBoxDetection}{$F$}
        \State Initialize list $L$ of $N$ tuples of type $(\code{parent}, \code{bbox})$, 
        \newline with $L[i].\code{parent} = i$ 
        \State $\code{ctr} = 1$
        \For{$i=1$ \textbf{to} $F.\code{height}$} 
        \For{$j=1$ \textbf{to} $F.\code{width}$} 
        \State $p = F[i][j]$
        \State $\code{isWhite} = (p \ne 0)$
        \State $(p_1, p_2, p_3, p_4) = \textsc{GetNeighbors}(F, i, j)$ 
        \State $(l_1, l_2, l_3, l_4) = \textsc{GetNeighborLabels}(F, i, j)$ 
        \State $\code{isNew} = (p_1 == p_2 == p_3 == p_4 == 0) \wedge \code{isWhite}$
        \State $l_\code{min} = \textsc{GetMinLabel}(l_1, l_2, l_3, l_4)$
        \State $l_\code{min} = \oselect(\code{isNew}, \code{ctr}, l_\code{min})$
        \State $\code{ctr} = \oselect(\lnot\code{isNew}, \code{ctr} + 1, \code{ctr})$
        \For{each label $l$ in $\{l_1, l_2, l_3, l_4\}$}
        \State $\textsc{UpdateParent}(L, \lnot\code{isNew}, l, l_\code{min})$
        \EndFor
        \State $\textsc{UpdateBBox}(L, \code{isWhite}, l_\code{min}, i, j)$
        \State $\textsc{SetLabel}(F, i, j, l_\code{min})$
        \EndFor
        \EndFor

        \Statex
        \For{$i=1$ \textbf{to} $N$}
        \State $\code{par} = L[i].\code{parent}$
        \State $\code{toMerge} = (\code{par} < i)$
        \For{$j=i$ \textbf{to} $1$}
        \State \begin{varwidth}[t]{\linewidth}
            $L[i].\code{parent} = \oselect(\code{toMerge}\wedge (\code{par}==j),$ \par
            \hskip\algorithmicindent $L[j].\code{parent}, L[i].\code{parent})$
        \end{varwidth}
        \EndFor
        \EndFor

        \For{$i=1$ \textbf{to} $N$}
        \For{$j=1$ \textbf{to} $N$}
        \State $\code{toMerge} = (L[j].\code{parent} == i)$
        \State $\textsc{MergeBBox}(\code{toMerge}, L[i].\code{bbox}, L[j].\code{bbox})$
        \EndFor
        \EndFor
        \State \Return $L$

        \EndProcedure
    \end{algorithmic}
\end{algorithm}

\begin{theorem}
    The bounding box detection algorithm in \cref{alg:bbox} is data-oblivious, with public parameters $N$, and the height and width of the input frame.
\end{theorem}
\begin{proof}
    The simulator generates a random frame $F$ of the given height and width and runs the algorithm.
    It outputs the trace produced by the algorithm.

    The access patterns of line 4 depends only on $N$.
    Line 5 has fixed access patterns.

    The loops (lines 6--22) are executed a fixed number of times, equal to the height and width of the frame.
    The access patterns of line 8 depend only on the loop variables $i$ and $j$, which are public information.
    Line 9 has fixed access patterns.
    In line 10, the function \textsc{GetNeighbors} returns the four pixels neighboring the input coordinates ($i$ and $j$), and its access patterns thus depend only on $i$ and $j$. 
    Similarly in line 11, \textsc{GetNeighborLabels} looks up the labels assigned to the neighboring pixels, and thus has access patterns that only depend on $i$ and $j$.
    Line 12 has fixed access patterns.
    In line 13, \textsc{GetMinLabel} selects the minimum of the input values using \oselect operations, and thus has fixed access patterns.
    Lines 14--15 have fixed access patterns.
    In line 17, \textsc{UpdateParent} uses \oarr combined with \oselect to update $L[l].\code{parent}$ to $l_\code{min}$; it thus has access patterns that only depend on the length $N$ of the array $L$.
    In line 18, \textsc{UpdateBbox} similarly uses \oarr combined with \oselect to update $L[l_\code{min}].\code{bbox}$ with the current coordinates $i$ and $j$; its access patterns therefore only depend on $L$'s length $N$.
    In line 19, \textsc{SetLabel} sets the label of the pixel at $F[i][j]$ to $l_\code{min}$; its access patterns depend only on the loop variables $i$ and $j$.

    Lines 24--28 are run $N$ times.
    The access patterns of line 24 depend only on the loop variable $i$.
    Line 25 has fixed access patterns.
    Line 27 is executed $i$ times, which is public information; also, the access patterns of this line only depend on the loop variables $i$ and $j$.

    Lines 32--33 are run $N^2$ times.
    The access patterns of line 32 only depend on the loop variable $j$.
    In line 33, the function \textsc{MergeBbox} uses \oselect operations to update $L[i].\code{bbox}$ with $L[j].\code{bbox}$; it therefore has fixed access patterns.

    Thus, the trace produced by the simulator is indistinguishable from the trace produced by a real run of the algorithm.
\end{proof}

\subsubsection{Object cropping}

\cref{alg:crop,alg:resize} together describe \sys's oblivious cropping algorithm.
\sys crops out images of a fixed upper bounded size using \cref{alg:crop}, and then scales up the ROI within the cropped image using \cref{alg:resize} (as described in \cref{s:algorithms:cropping}).
The pseudocode is self-explanatory.

\begin{algorithm}[t]
    \small
    \caption{Object cropping}\label{alg:crop}
    \begin{algorithmic}[1]
        \State \textbf{Constants:} Upper bounds on object dimensions \code{height}, \code{width}
        \State \textbf{Input:} Frame $F$, bounding box coordinates \code{bbox}
        \Procedure{CropObject}{$F$, \code{bbox}}
        \State Initialize an empty buffer \code{buf} with width $=F.\code{width}$ and height $=\code{height}$
        \For{$i=1$ \textbf{to} $F$.\code{height}}
        \State $\code{cond} = (i == \code{bbox.top})$
        \State $\textsc{CopyRows}(\code{cond}, i, F, \code{buf})$
        \EndFor

        \Statex
        \State Initialize an empty buffer \code{obj} with width $=\code{width}$ and height $=\code{height}$
        \For{$i=1$ \textbf{to} $F$.\code{width}}
        \State $\code{cond} = (i == \code{bbox.left})$
        \State $\textsc{CopyCols}(\code{cond}, i, \code{buf}, \code{obj})$
        \EndFor
        \EndProcedure
    \end{algorithmic}
\end{algorithm}

\begin{theorem}
\label{thm:cropping}
    The object cropping algorithm in \cref{alg:crop} is data-oblivious, with public parameters equal to the dimensions of the input frame, and the upper bounds on the object dimensions \code{height} and \code{width}.
\end{theorem}

\begin{proof}
    The simulator generates a random frame of the given dimensions, along with a bounding box \code{bbox} with random coordinates.
    It then runs the algorithm, and outputs the produced trace.

    The access patterns of line 4 depend only the frame's width, and the parameter \code{width}, both of which are known to the simulator.
    Lines 6--7 run a fixed number of times, equal to the height of the frame.
    Line 6 has fixed access patterns.
    In line 7, \textsc{CopyRows} uses \oselect to copy pixels from $F$ into \code{buf}; its access patterns thus only depend on the loop variable $i$, the width of the frame, and the parameter \code{height}.

    The access patterns of line 9 depend only the parameters \code{width} and \code{height}.
    Lines 11--12 run a fixed number of times, equal to the width of the frame.
    Line 11 has fixed access patterns.
    In line 12, \textsc{CopyCols} uses \oselect to copy pixels from \code{buf} into \code{obj}; its access patterns thus only depend on the loop variable $i$, and the parameters \code{height} and \code{width}.

    Thus, the trace produced by the simulator is indistinguishable from the trace produced by a real run of the algorithm.
\end{proof}

\begin{algorithm}[t!]
    \small
    \caption{Object resizing}\label{alg:resize}
    \begin{algorithmic}[1]
        \State \textbf{Input:} Object buffer $O$, bounding box coordinates \code{bbox}
        \Procedure{ResizeObject}{$O$, \code{bbox}}
        \State\textsc{ResizeHorizontally}({$O$, \code{bbox}, \code{false}})
        \State\textsc{Transpose}($O$)
        \State\textsc{ResizeHorizontally}({$O$, \code{bbox}, \code{true}})
        \State\textsc{Transpose}($O$)
        \EndProcedure
        
        \Statex
        \Procedure{ResizeHorizontally}{$O$, \code{bbox}}%
        \For{$i=1$ \textbf{to} $O$.\code{height}}
        \For{$j=1$ \textbf{to} $O$.\code{width}}
        \State $p = \textsc{PixelOfInterest}(j, \code{bbox})$%
        \State $a = \oarr($O[i]$, p)$
        \State $b = \oarr($O[i]$, p+1)$
        \State $O[i][j] = \textsc{LinearInterpolate}(a, b)$
        \EndFor
        \EndFor
        \EndProcedure
        
    \end{algorithmic}
\end{algorithm}

\begin{theorem}
    The object resizing algorithm in \cref{alg:resize} is data-oblivious, with public parameters equal to the dimensions of the input object $O$.
\end{theorem}

\begin{proof}
    The simulator generates a random object buffer with the given dimensions, along with a bounding box \code{bbox} with random coordinates.
    It then runs the algorithm, and outputs the produced trace.

    The function \code{Transpose} transposes the object buffer, and thus its access patterns only depend on the dimensions of $O$.
    The function \code{ResizeHorizontally} works as follows.
    The loops (lines 9--15) are executed a fixed number of times, equal to the dimensions of $O$. 
    Line 11 computes the location of the pixels to be used for linearly interpolating the current pixel, using a set of arithmetic operations; it thus has fixed access patterns.
    Lines 12 and 13 have access patterns that only depend on the width of $O$.
    In line 14, \textsc{LinearInterpolate} linearly interpolates the current pixel's value using a set of arithmetic operations; the access patterns of this line thus depend only on the loop variables.

    Thus, the trace produced by the simulator is indistinguishable from the trace produced by a real run of the algorithm.
\end{proof}

\subsubsection{Object tracking}
\cref{alg:featuredet} describes the feature detection phase of the object tracking. We omit a description of feature matching since it is oblivious by default.

The algorithm first creates a set of increasingly blurred versions of the input image (line 5).
Then, it identifies a set of candidate \emph{keypoints} in these blurred images, \ie pixels that are the maximum and minimum of all their neighbors (lines 6--14).
This set of keypoints is further refined to identify those that are robust to changes in illumination (\ie have high intensity), or represent a ``corner'' (lines 15--18).
Mathematically, these tests involve the computation of derivatives at the candidate point, and then a comparison of the results against a threshold. Candidates that fail these tests are discarded.

Finally, for each keypoint, the algorithm constructs a \emph{feature descriptor}.
It calculates the ``orientation'' of the pixels around the keypoint (within a $16\times16$ neighborhood) based on pixel differences, and then constructs a histogram over the computed values (lines 20--14).
The histogram acts as the keypoint's descriptor.

\begin{algorithm}[t!]
    \small
    \caption{Feature detection}\label{alg:featuredet}
    \begin{algorithmic}[1]
        \State \textbf{Input:} Object buffer $O$, maximum number of candidate keypoints $N_\code{temp}$, maximum number of actual keypoints $N$
        \Procedure{DetectFeatures}{$O$, $N_\code{temp}$, $N$}
        \State Initialize an empty list $L$ of size $N_\code{temp}$ for candidate keypoints, and a list $H$ of size $N$ for features of final keypoints
        \State $\code{ctr} = 0$
        \State $\code{images} = \textsc{GetDifferenceOfGaussians}(O)$
        \For {each pixel $p$ in \code{images}}
        \State $\code{nbrs} = \textsc{GetNeighbors}(p)$
        \State $\code{isExtrema} = \textsc{CheckExtrema}(p, \code{nbrs})$
        \State $k = (p, \code{nbrs})$
        \For {$i = 1$ to $N_\code{temp}$}
        \State $L[i] = \oselect(\code{isExtrema} \wedge i == \code{ctr}, k, L[i])$
        \EndFor
        \State $\code{ctr} = \oselect(\code{isExtrema} \wedge \code{ctr} < N_\code{temp}, \code{ctr+1}, \code{ctr})$
        \EndFor
        \Statex
        \For {$i = 1$ to $N_\code{temp}$}
        \State $\code{isRobust} = \textsc{CheckRobustness}(L[i])$
        \State $L[i] = \oselect(\code{isRobust}, L[i], \mynull)$
        \EndFor
        \State $\osort(L)$ such that non-null values move to the head of $L$
        \Statex
        \For {$i = 1$ to $N$}
        \State $\code{bbox} = \textsc{CalcNeighborhoodBbox}(L[i])$
        \State $\code{roi} = \textsc{CropObject}(\code{images}, \code{bbox})$
        \State $H[i] = \textsc{CalcOrientationHist}(L[i], \code{roi})$
        \EndFor
        \State \Return $H$
        \EndProcedure
        
    \end{algorithmic}
\end{algorithm}

\begin{theorem}
    The feature detection algorithm in \cref{alg:featuredet} is data-oblivious, with public parameters equal to the dimensions of the input image $O$, and upper bounds $N_\code{temp}$ and $N$.
\end{theorem}

\begin{proof}
     The simulator generates a random image buffer with the given dimensions, and then runs the algorithm.
     It outputs the trace produced by the algorithm.
     
     Line 5 performs Gaussian blurring operations on the input image $O$, which perform a convolution of the input image with a specified \emph{kernel} (\ie a small matrix). The access patterns of these matrix multiplications are fixed, and independent of the values of the matrices.
     
     The loop (lines 6--14) runs a fixed number of times, the value of which depends on the resolution of the input image, which is public.
     Line 7 fetches the neighbors of the current pixel; its access patterns are therefore dependent only on coordinates of the loop variable, which is public.
     Line 8 checks the value of the current pixel with the obtained \code{nbrs} using \oselect operations, and thus has fixed access patterns, independent of the values.
     Line 9 has fixed access patterns.
     The loop in lines 10--12 executes a fixed $N_\code{temp}$ number of times.
     The access patterns of line 11 depend only on the public loop variable.
     Line 13 has fixed access patterns.
     
     The loop in lines 15--18 executes a fixed $N_\code{temp}$ number of times.
     Line 16 has fixed access patterns.
     The access patterns of line 17 depend only on the public loop variable.
     Line 19 has fixed access patterns that depend only on the size $N_\code{temp}$ of the array $L$.

     The loop in lines 15--18 executes a fixed $N$ number of times.
     The function \textsc{CalcNeighborhoodBbox} in Line 21 computes the bounding box of the $16\times16$ neighborhood of the current keypoint using arithmetic operations, and has fixed access patterns.
     The access patterns of the function \textsc{CropObject} depend only on the dimensions of the input image $O$ (which is public) and the resolution of the bounding box, which is fixed (from \cref{thm:cropping}).
     The function \textsc{CalcOrientationHist} performs arithmetic operations, and the access patterns of line 23 depend only on the public loop variable.
     
     Thus, the trace produced by the simulator is indistinguishable from the trace produced by a real run of the algorithm.
\end{proof}

\fi

\section{Impact of Video Encoder Padding} \label{s:appendix:padding}
In \sys, the source video streams are padded at the camera to prevent information leakage  due to variations in bitrate of the encrypted network traffic.
However, it may not always be possible to modify legacy cameras to incorporate padding. This security guarantee also comes at the cost of performance and increased network bandwidth.

While we recommend padding the video streams for security, we studied the impact of disabling video encoder padding on \sys so as to aid practitioners in taking an informed decision between security and performance. Disabling padding has two  implications on \sys.

First, the encoded stream may also contain interframes in addition to keyframes (see \cref{s:padding}). Thus, we  have devised an oblivious routine for interframe prediction, which is described in \cref{s:decoding:inter}.
Second, the performance overhead of \sys (\approx\finaloverhead) reduces to a range of \approx\finaloverheadnopadding. This is due to lower interframe decoding latency and smaller number of decoded bits per row of blocks (which are obliviously sorted). %

\subsection{Inter-Prediction for Interframes} \label{s:decoding:inter}
Inter-predicted blocks use {\em previously decoded frames} as reference (either the previous frame, or the most recent keyframe). %
Obliviousness of inter-prediction requires that the reference block (which frame, and block's coordinates therein) remains private during decoding. Otherwise, an attacker observing access patterns during inter-prediction can discern the motion of objects across frames. Furthermore, some blocks even in interframes can be \emph{intra}-predicted for coding efficiency, and oblivious approaches need to conceal whether an interframe block is inter- or intra-predicted. 
A na\"ive, but inefficient, approach to achieve obliviousness is to access \emph{all blocks in possible reference frames} at least once---if any block is left untouched, its location its leaked to the attacker. %

We leverage properties of video streams to make our oblivious solution efficient: 
\begin{enumerate*}[($i$)]
	\item Most blocks in interframes are inter-predicted (\approx$99\%$ blocks in our streams); and 
	\item Coordinates of reference blocks are close to the coordinates of inter-predicted blocks (in a previous frame), \eg $90\%$ of blocks are radially within 1 to 3 blocks.
\end{enumerate*}
These properties enable two optimizations. %
First, we assume every block in an interframe is inter-predicted.
Any error due to this assumption on intra-predicted blocks is minor in practice. 
Second, %
instead of scanning all blocks in prior frames, we only access blocks within a small distance of the current block.
If the reference block is indeed within this distance, we fetch it obliviously using \oarr; else, (in the rare cases) we use the block at the same coordinates in the previous frame as reference.

\section{Impact of Disabling Hyperthreading}\label{s:hyperthreading}
\sys requires hyperthreading to be disabled in the underlying system for security (see \cref{s:threatmodel}). In contrast, in our evaluation, the baseline system leveraged hyperthreading to maximize its throughput. 

We measured the impact of disabling hyperthreading on \sys's performance to be $5\%$. Visor heavily utilizes vector units due to the increased data-level parallelism of oblivious algorithms, leaving little space for performance improvement when hyperthreading is enabled~\cite{DCDRF:SOCC18}. As such, the increased security comes with negligible performance overhead.

Disabling hyperthreading in cloud VMs is considered to be a good practice due to the reduced impact of microarchitectural data-sampling vulnerabilities that affect commodity Intel CPUs (not just Intel SGX)~\cite{MDS:attack:RIDL,MDS:attack:Fallout,SGX:attack:ZombieLoad,cacheOut}. Our experiments demonstrate that disabling hyperthreading in the baseline system reduces its performance by $30\%$; bridging considerably the performance gap between Visor and insecure baseline systems in hyperthreading-disabled cloud deployments.

\end{document}